\documentclass[11pt]{article}

\usepackage[margin=1in]{geometry}
\usepackage{amsmath,amssymb,amsthm,mathtools}
\usepackage{booktabs}
\usepackage{tabularx}
\usepackage{graphicx}
\graphicspath{{figures/}}
\usepackage{microtype}
\usepackage{enumitem}
\usepackage[round]{natbib}
\usepackage{xcolor}
\usepackage[colorlinks=true,linkcolor=blue!50!black,citecolor=blue!50!black,urlcolor=blue!50!black]{hyperref}
\hypersetup{
  pdftitle={The Mathematics of Volatility Surfaces},
  pdfauthor={Miquel Noguer i Alonso},
  pdfsubject={Arbitrage geometry, local volatility, stochastic fields, neural operators, normalizing flows, and signatures},
  pdfkeywords={volatility surface, local volatility, static arbitrage, stochastic fields, neural operators, normalizing flows, signatures}
}

\newtheorem{theorem}{Theorem}[section]
\newtheorem{proposition}[theorem]{Proposition}
\newtheorem{lemma}[theorem]{Lemma}
\newtheorem{corollary}[theorem]{Corollary}

\theoremstyle{definition}
\newtheorem{definition}[theorem]{Definition}
\newtheorem{assumption}[theorem]{Assumption}

\theoremstyle{remark}
\newtheorem{remark}[theorem]{Remark}

\newcommand{\tier}[1]{\textnormal{[\textsc{#1}]}}

\newcommand{\R}{\mathbb{R}}
\newcommand{\E}{\mathbb{E}}
\newcommand{\Prob}{\mathbb{P}}
\newcommand{\Hs}{\mathcal{H}}
\newcommand{\Ksurf}{\mathcal{K}}
\newcommand{\Kc}{\mathcal{K}_{c}}
\newcommand{\Mcap}{\overline M}
\newcommand{\Kcmbar}{\mathcal{K}_{c,m,\Mcap}}
\newcommand{\Ka}{\mathcal{K}_{a}}
\newcommand{\Dom}{\mathcal{D}}
\newcommand{\tr}{\operatorname{tr}}
\newcommand{\ip}[2]{\left\langle #1,\, #2 \right\rangle}
\newcommand{\norm}[1]{\left\lVert #1 \right\rVert}
\newcommand{\dd}{\,\mathrm{d}}

\title{\textbf{The Mathematics of Volatility Surfaces}\\[0.4em]
\large Arbitrage Geometry, Local Volatility, Stochastic Fields,\\
Neural Operators, Normalizing Flows, and Signatures}

\author{Miquel Noguer i Alonso\\
\normalsize Artificial Intelligence Finance Institute (AIFI)}

\date{\today}

\begin{document}

\maketitle

\begin{abstract}
\noindent This paper develops a unified mathematical theory of implied, local, and learned volatility surfaces. Total variance $w_t(k,\tau)=\tau\sigma_t^2(k,\tau)$ is an infinite-dimensional state constrained by positivity, calendar monotonicity, and the butterfly differential inequality. We establish the topology and tangent geometry of this arbitrage set and prove that a nondegenerate Gaussian shock at an active constraint exits with probability tending to one half. Exact invariance therefore requires tangency, reflection, or confinement to an arbitrage-free manifold. We separate this static invariance problem from dynamic no-arbitrage, derive the Musiela maturity-transport identity, and identify the additional fixed-contract martingale restriction. We formulate Hilbert-space dynamics, prove an exact modal reduction with closed-form truncation error, derive Karhunen--Lo\`eve factors, identify the portfolio derivative as a vega field, and obtain the covariance-optimal hedge $\alpha^\ast=(H^\ast C H)^{-1}H^\ast C\nu$. The local-volatility chart completes the geometry: Dupire local variance is the ratio $a=\partial_\tau w/g[w]$ of the calendar and butterfly constraint functionals. Neural operators provide arbitrage-free universal approximation through simplex and cone heads. Normalizing-flow maps then add tractable conditional densities: we derive exact change-of-variables formulae for exponential local-variance flows and invertible stick-breaking price-simplex flows, while stating the quasi-invariance conditions required in genuine function space. Finally, fading-signature fields encode surface history and yield autonomous finite-dimensional controlled dynamics. The result is one framework for representation, dynamics, arbitrage, dimension reduction, likelihood-based learning, simulation, and hedging, together with a falsifiable empirical protocol.
\medskip


\end{abstract}
\newpage
\tableofcontents
\newpage
\section{Introduction}\label{sec:intro}

The implied volatility surface is the market's state variable for index options, yet most of the modeling literature does not treat it as a state. Parametric families fit the surface slice by slice; stochastic volatility models generate it as an output of a low-dimensional latent process; desk practice differentiates it into buckets. In each case the object that the market actually quotes --- a two-dimensional field with hard shape constraints --- is dismantled before the mathematics begins.

This paper keeps the object whole. We model implied total variance
\begin{equation}\label{eq:state}
w_t(k,\tau) \;=\; \tau\,\sigma_t^2(k,\tau), \qquad (k,\tau)\in\Dom \coloneqq [k_{\min},k_{\max}]\times[\tau_{\min},\tau_{\max}],
\end{equation}
as a stochastic process taking values in a weighted Sobolev space $\Hs=H^2_\rho(\Dom)$, and we take seriously the fact that at every time $t$ the realized field must lie in the admissible set
\footnote{This manuscript consolidates, substantially revises, and supersedes
three 2026 working drafts by the author: \emph{Volatility Surfaces as
Constrained Stochastic Fields}, \emph{Random Fields as Multidimensional
Stochastic Differential Equations}, and \emph{Local Volatility Fields}. The
present paper is the archival statement of the combined results; overlapping
drafts should not be treated as independent contributions.}
\begin{equation}\label{eq:Kintro}
\Ksurf \;=\; \bigl\{\, w\in\Hs \;:\; w>0,\;\; \partial_\tau w \ge 0,\;\; g[w]\ge 0 \,\bigr\},
\end{equation}
where $g[w]$ is the Gatheral--Jacquier butterfly functional recalled in Section~\ref{sec:constraints}. The surface is thus a \emph{constrained stochastic field}: a random element of an infinite-dimensional space that never leaves a thin, curved, boundary-active subset of it. (We use ``stochastic field'' and ``random field'' interchangeably throughout.)

Two questions drive the paper. First, \emph{what does the constraint geometry do to the dynamics?} We show that the answer is drastic: Gaussian field models --- the natural infinite-dimensional analogue of the PCA-driven factor models used since \citet{cont2002dynamics} --- are incompatible with the boundary of $\Ksurf$. At any surface with an active constraint, a nondegenerate Gaussian shock exits the admissible set with probability tending to $1/2$ as the time step shrinks (Theorem~\ref{thm:boundary}), and in continuous time exact invariance forces degeneracy, reflection, or confinement to a submanifold (Corollary~\ref{cor:trichotomy}). This is the surface-space analogue of the square-root boundary behavior familiar from short-rate models, and of the consistency results of \citet{filipovic2001consistency} for Heath--Jarrow--Morton models: shape constraints are not decoration, they are the binding physics of the state space.

Second, \emph{what does the field perspective buy the practitioner?} Three things. (i) A spectral theory: the increment covariance operator is trace class, so Karhunen--Lo\`eve factors exist, and the familiar level/skew/curvature modes of surface PCA become rigorous objects with quantified truncation error (Section~\ref{sec:spectral}). (ii) A hedging theory: the sensitivity of a book to the surface is a single object, the \emph{vega field} $\nu_w\in\Hs$, defined as the Riesz representer of the Fr\'echet derivative of the portfolio value functional; bucketed vegas are samples of this field under reproducing-kernel bumps (Lemma~\ref{lem:bucket}), and the minimum-variance hedge with $n$ traded options has the closed form $\alpha^\ast=G^{-1}H^\ast C\nu$ (Theorem~\ref{thm:hedge}). (iii) A model-design discipline: the boundary theorem sorts the space of possible models into three viable construction templates --- convex price-coordinate dynamics, tangent-projected innovations, and arbitrage-free parameter maps --- which we develop and compare in Section~\ref{sec:constructions}.

\subsection{Contributions}\label{sec:contributions}

\begin{enumerate}[label=\textbf{C\arabic*.},leftmargin=2.6em]
\item \textbf{State space.} A formulation of the surface as an $H^2_\rho(\Dom)$-valued process, with the embedding, reproducing-kernel, and price-coordinate structure needed downstream (Section~\ref{sec:statespace}).
\item \textbf{Constraint geometry.} Closedness and weak closedness of the admissible set with a positive variance floor; an exact banded correspondence between variance and price charts; convexity in price coordinates versus nonconvexity in total-variance coordinates; the contingent tangent cone and an explicit directional derivative of the butterfly functional; a first-order feasibility criterion (Sections~\ref{sec:constraints}--\ref{sec:geometry}).
\item \textbf{Boundary incompatibility.} The $1/2$-limit theorem for Euler shocks of nondegenerate Gaussian field dynamics at active constraints, a uniform $\sqrt\Delta$ boundary-layer law, and the continuous-time trichotomy (degenerate normal noise, reflection, or submanifold confinement) with an explicit localized DDS proof (Section~\ref{sec:dynamics}).
\item \textbf{Interior survival.} A Borell--TIS estimate quantifying when unconstrained Gaussian steps are safe: an explicit exponential bound in terms of the constraint margin and the $C^2$ concentration of the noise (Section~\ref{sec:dynamics}).
\item \textbf{Infinite-dimensional dynamics and exact reduction.} A Hilbert-space SDE formulation and an exact diagonal modal theorem that identifies the finite-dimensional factor SDE, the function-valued existence condition, and the full mean-square truncation error (Section~\ref{sec:hilbertdynamics}).
\item \textbf{Spectral dynamics and functional hedging.} Karhunen--Lo\`eve factors of the increment covariance operator; the vega field as a Riesz representer; an exact covariance law under changes of Sobolev weight; the closed-form minimum-variance field hedge $\alpha^\ast = G^{-1}H^\ast C\nu$ with its residual-risk formula and a joint spot--field extension (Sections~\ref{sec:spectral}--\ref{sec:hedging}).
\item \textbf{Local-volatility chart.} The identity $a=\partial_\tau w/g[w]$, its zero--pole geometry, the convex-cone chart, a stable parabolic inverse, and the SSVI at-the-money local-variance formula (Section~\ref{sec:localchart}).
\item \textbf{Dynamic-arbitrage separation.} An exact Musiela roll-down identity, its cancellation along a fixed-expiry contract, and a precise separation between static surface feasibility and the risk-neutral drift restriction for discounted option prices (Section~\ref{sec:dynamicarb}).
\item \textbf{Learning, likelihood, and memory.} Constraint-preserving neural-operator heads, exact normalizing-flow maps in the local and price charts, and fading-signature fields giving finite-dimensional path-dependent realizations with separate surface and history truncation axes (Sections~\ref{sec:neuraloperators}--\ref{sec:signaturefields}).
\item \textbf{Constructions and a pre-registered empirical program.} Classical, learned, and path-dependent model templates with tiered guarantees, including a fully well-posed finite-grid reflected price construction, and a falsifiable end-of-day SPX design: nine models, six metric families, hypotheses H1--H7, and a reproducibility checklist (Sections~\ref{sec:constructions}--\ref{sec:empirics}).
\end{enumerate}

\subsection{Related literature}\label{sec:literature}

\emph{Dynamics of implied volatility.} The empirical program descends from \citet{skiadopoulos1999dynamics} and \citet{cont2002dynamics}, who applied principal component analysis to smiles and surfaces and documented the small number of dominant modes. Our Section~\ref{sec:spectral} is the operator-theoretic completion of that program; our Section~\ref{sec:dynamics} explains why its Gaussian-factor reading cannot be exactly arbitrage-free.

\emph{Market models of option prices.} Direct modeling of implied volatilities or option prices as states goes back to \citet{schonbucher1999market}, with the existence and drift-restriction theory developed by \citet{schweizer2008term,schweizer2008multistrike}, \citet{carmona2009local} for local-volatility codebooks, and \citet{kallsen2015hjm} in a Heath--Jarrow--Morton spirit; \citet{durrleman2010implied} connects implied and spot volatilities. Relative to this line, our contribution is to work on the full two-dimensional field with the static-arbitrage set treated as a geometric object --- tangent cones, normal directions, projections --- and then to state exactly where this geometry stops and the dynamic drift restriction begins.

\emph{Static arbitrage.} The constraint set itself is classical: \citet{carr2005note}, \citet{davis2007range}, and \citet{roper2010arbitrage} characterize arbitrage-free call-price and total-variance surfaces; \citet{gatheral2014arbitrage} give the butterfly functional and the SVI/SSVI parametrizations we use as the smooth-map construction; \citet{lee2004moment} bounds wing growth; \citet{fengler2009smoothing} and \citet{aitsahalia2003nonparametric} develop shape-constrained estimation; \citet{cohen2020repair} repair quoted surfaces. We import these results as the definition of $\Ksurf$ and study their interaction with stochastic dynamics.

\emph{Local volatility and realizations.} Dupire's local-volatility construction is \citet{dupire1994pricing}; the total-variance representation and SVI calculus are standard since \citet{gatheral2006surface}. Finite-dimensional realization theory follows \citet{bjork2001existence} and \citet{filipovic2001consistency}; the boundary-degenerate benchmark is the square-root model of \citet{cox1985theory}.

\emph{Infinite-dimensional stochastic analysis and random fields.} The functional-analytic toolkit is from \citet{daprato2014stochastic} (Hilbert-space-valued diffusions), \citet{aubin1990setvalued} and \citet{rockafellar1998variational} (tangent cones and set-valued analysis), \citet{adler2007random} and \citet{borell1975brunn} (Gaussian concentration on Banach spaces), and \citet{bogachev1998gaussian}. Reflected constructions in function space connect to \citet{nualart1992reflection}, \citet{haussmann1989stochastic}, and, in finite dimensions, \citet{cepa1998skorohod} and \citet{slominski2001euler}. Functional data analysis enters through \citet{ramsay2005functional} and \citet{horvath2012inference}.

\emph{Modern surface generators.} Neural and generative surface models --- \citet{cohen2023neural}, \citet{cont2023simulation}, \citet{cuchiero2020gan} --- enforce or encourage no-arbitrage inside flexible function classes; rough- and path-dependent volatility \citep{gatheral2018rough,guyon2023path} constrains what realistic dynamics must reproduce. Our framework supplies the geometry against which all such generators can be audited, and our empirical program (Section~\ref{sec:empirics}) includes them as competing models.

\subsection{Claim tiering and standards}\label{sec:tiering}

Following the standards used across this program, every formal statement is tagged with one of three tiers. \tier{Proved} means a complete proof is given here or in the cited source under the stated hypotheses. \tier{Conditional} means the statement is proved modulo an explicitly identified hypothesis or verification that we have not carried out (for example, the verification that a cited well-posedness theorem applies to our specific constraint set). \tier{Conjectural} means we state the claim with supporting reasoning but without proof. No numerical results are reported in this version; Section~\ref{sec:empirics} is a pre-registered design, and every threshold appearing there is a design constant fixed before data contact, not an estimate.

\medskip
\noindent\emph{Roadmap.} Sections~\ref{sec:statespace}--\ref{sec:geometry} build the state space and its arbitrage geometry. Section~\ref{sec:localchart} develops the local-volatility chart, and Section~\ref{sec:dynamicarb} separates static feasibility from dynamic no-arbitrage through the Musiela coordinate map. Section~\ref{sec:dynamics} proves the boundary theorem and interior survival estimate. Sections~\ref{sec:hilbertdynamics}--\ref{sec:hedging} develop dynamics, spectral reduction, and hedging. Section~\ref{sec:constructions} gives classical constructions; Sections~\ref{sec:neuraloperators}--\ref{sec:signaturefields} add operator learning, likelihood flows, and path-dependent realizations. Section~\ref{sec:empirics} specifies the empirical program. Section~\ref{sec:scope} records the theorem frontier. Appendices contain technical derivations and the reproducibility checklist.

\section{The state space: total variance fields in a weighted Sobolev space}\label{sec:statespace}

\subsection{Domain and coordinates}\label{sec:domain}

Fix a forward curve $F_t(\tau)$ and write $k=\log(K/F_t(\tau))$ for forward log-moneyness. The modeling domain is the compact rectangle $\Dom=[k_{\min},k_{\max}]\times[\tau_{\min},\tau_{\max}]$ with $\tau_{\min}>0$ and $-\infty<k_{\min}<0<k_{\max}<\infty$. Working on a compact domain with $\tau_{\min}>0$ is an economic choice, not a technical evasion: quoted surfaces are supported on such a window, the wing asymptotics of \citet{lee2004moment} govern extrapolation beyond it, and every statement in this paper is about the quoted window. The state variable is the total variance field \eqref{eq:state}; we suppress $t$ when statics are in view and write $w_k=\partial_k w$, $w_{kk}=\partial^2_k w$, $w_\tau=\partial_\tau w$.

\subsection{The space \texorpdfstring{$H^2_\rho(\Dom)$}{H2 rho(D)}}\label{sec:sobolev}

Let $\rho\in C(\overline{\Dom})$ satisfy $0<\rho_{\min}\le\rho\le\rho_{\max}<\infty$. Define
\begin{equation}\label{eq:norm}
\ip{u}{v}_{\Hs} \;=\; \sum_{|\beta|\le 2} \int_{\Dom} \partial^\beta u\,(x)\; \partial^\beta v\,(x)\; \rho(x) \dd x,
\qquad \Hs \;=\; H^2_\rho(\Dom) \;=\; \bigl(H^2(\Dom),\ip{\cdot}{\cdot}_{\Hs}\bigr).
\end{equation}
Because $\rho$ is bounded above and below, the $\Hs$-norm is equivalent to the standard $H^2$-norm: the weight changes the geometry (adjoints, Riesz representers, orthogonality) but not the topology. The weight is where market structure enters the inner product --- vega-weighting, liquidity-weighting, or maturity discounting are all choices of $\rho$ --- and Lemma~\ref{lem:weightcovariance} makes the covariance exact while showing that the optimal hedge itself is weight-invariant.

\begin{proposition}[Embedding and reproducing kernel; \tier{Proved}]\label{prop:embedding}
Let $\Dom\subset\R^2$ be a compact rectangle. Then:
\begin{enumerate}[label=(\roman*)]
\item $H^2(\Dom)\hookrightarrow C^{0,\alpha}(\overline{\Dom})$ continuously for every $\alpha\in(0,1)$, and compactly into $C^0(\overline{\Dom})$; moreover $H^2(\Dom)\hookrightarrow\hookrightarrow H^s(\Dom)$ compactly for every $s<2$.
\item Point evaluation $\delta_x:w\mapsto w(x)$ is a bounded linear functional on $\Hs$ for every $x\in\overline{\Dom}$; hence $\Hs$ is a reproducing kernel Hilbert space: there exist kernel sections $R_x\in\Hs$ with $w(x)=\ip{R_x}{w}_{\Hs}$ for all $w\in\Hs$.
\item First derivatives evaluate weakly but not pointwise: $\partial^\beta w\in H^1(\Dom)\subset L^p(\Dom)$ for all $p<\infty$ when $|\beta|=1$, and $\partial^\beta w\in L^2(\Dom)$ when $|\beta|=2$.
\end{enumerate}
\end{proposition}

\begin{proof}
(i) is the Sobolev embedding theorem in dimension $d=2$ with $s=2>d/2$, plus Rellich--Kondrachov compactness on the Lipschitz domain $\Dom$. (ii) follows from (i): $|w(x)|\le \norm{w}_{C^0}\le c\norm{w}_{H^2}\le c'\norm{w}_{\Hs}$, and Riesz representation yields $R_x$. (iii) is again Sobolev embedding, now for $H^1(\Dom)$ in $d=2$.
\end{proof}

Part (iii) is the honest print on everything that follows: statements about $w$ itself can be made pointwise, statements about $w_k$ can be made in every $L^p$, and statements about $w_{kk}$ --- hence about the butterfly functional --- can only be made almost everywhere or after pairing with a test function. The paper keeps this discipline explicit rather than assuming smoothness it does not need.

\subsection{Price coordinates and the Black--Scholes map}\label{sec:pricecoords}

For $w>0$ define the normalized Black--Scholes call price at $(k,\tau)$ by
\begin{equation}\label{eq:bs}
c \;=\; B(k,w) \;\coloneqq\; N(d_+) - e^{k} N(d_-), \qquad d_\pm \;=\; -\frac{k}{\sqrt{w}} \pm \frac{\sqrt{w}}{2},
\end{equation}
where $N$ is the standard normal distribution function; $c$ is the forward-discounted call price divided by the forward. For each fixed $k$, $w\mapsto B(k,w)$ is smooth and strictly increasing,
\begin{equation}\label{eq:vega}
\partial_w B(k,w) \;=\; \frac{n(d_+)}{2\sqrt{w}} \;>\; 0,
\end{equation}
with $n=N'$, so $B(k,\cdot)$ is a diffeomorphism from $(0,\infty)$ onto its image $\bigl((1-e^k)_+,\,1\bigr)$. Composing pointwise, $B$ maps total-variance fields to normalized call-price fields; we write $c=B[w]$ and $w=B^{-1}[c]$ for the induced maps on functions and record that both preserve $H^2$-regularity on sets where $w$ is bounded away from $0$ and $\infty$ (all coefficient functions of the chain rule are then smooth and bounded). The two coordinate systems carry the same information; Section~\ref{sec:constraints} shows they carry very different geometry.

\section{Static no-arbitrage as a constraint set}\label{sec:constraints}

\subsection{The three constraint families}\label{sec:families}

Static arbitrage on the quoted window is excluded by three families of shape constraints on $w$ \citep{carr2005note,davis2007range,roper2010arbitrage,gatheral2014arbitrage}.

\paragraph{Positivity.} $w(k,\tau)>0$ on $\Dom$: total implied variance is positive wherever a vol is quoted.

\paragraph{Calendar monotonicity.} At fixed forward log-moneyness,
\begin{equation}\label{eq:calendar}
\partial_\tau w(k,\tau)\;\ge\;0 \qquad \text{a.e.\ on } \Dom.
\end{equation}
This is equivalent to the absence of calendar-spread arbitrage between maturities when strikes are compared at equal forward moneyness \citep[Lem.~2.1]{gatheral2014arbitrage}; the fixed-strike statement differs by the forward drift, a distinction that matters in data construction (Section~\ref{sec:empirics}) but not in the geometry.

\paragraph{Butterfly positivity.} For each fixed $\tau$, define the Gatheral--Jacquier functional
\begin{equation}\label{eq:gfun}
g[w](k,\tau) \;=\; \Bigl(1-\frac{k\,w_k}{2w}\Bigr)^{\!2} \;-\; \frac{w_k^2}{4}\Bigl(\frac{1}{w}+\frac14\Bigr) \;+\; \frac{w_{kk}}{2}.
\end{equation}
On slices with enough regularity, the implied risk-neutral density of $\log(S_\tau/F)$ is
\begin{equation}\label{eq:density}
p(k,\tau) \;=\; \frac{g[w](k,\tau)}{\sqrt{2\pi\, w(k,\tau)}}\, \exp\!\Bigl(-\tfrac12 d_-(k,w(k,\tau))^2\Bigr),
\end{equation}
so $g[w]\ge 0$ is equivalent to butterfly-arbitrage-freeness of the slice \citep[Sec.~2]{gatheral2014arbitrage}. For $w\in\Hs$ with $w\ge m>0$, the right-hand side of \eqref{eq:gfun} is well defined as an element of $L^2(\Dom)$: the term $w_{kk}/2$ is $L^2$ by definition of $\Hs$, and the lower-order terms lie in $L^2$ by Proposition~\ref{prop:embedding}(iii) and the uniform bounds $m\le w\le \norm{w}_{C^0}$. The constraint $g[w]\ge 0$ is imposed almost everywhere. The ratio $\partial_\tau w/g[w]$ is Dupire's local variance wherever the denominator is positive, which makes the same two constraint functionals dynamically observable through the local-volatility chart.

\paragraph{Wings.} The moment bounds of \citet{lee2004moment} constrain $w_k$ as $|k|\to\infty$. On the compact window $\Dom$ they are dominated by the butterfly constraint and are omitted from $\Ksurf$; they reappear as extrapolation discipline in Section~\ref{sec:empirics}.

\subsection{The admissible set and its topology}\label{sec:admissible}

Because positivity with strict inequality does not define a closed set, we work throughout with a variance floor $m>0$ and define
\begin{equation}\label{eq:Km}
\Ksurf_m \;=\; \bigl\{\, w\in\Hs \;:\; w\ge m \text{ on } \overline{\Dom},\;\; \partial_\tau w\ge 0 \text{ a.e.},\;\; g[w]\ge 0 \text{ a.e.} \,\bigr\},
\qquad \Ksurf \;=\; \bigcup_{m>0}\Ksurf_m .
\end{equation}
The floor is economically innocuous on $\Dom$: with $\tau_{\min}>0$ and any positive lower bound on quoted vols, realized surfaces satisfy $w\ge m$ for some $m>0$. All analytic statements below are for $\Ksurf_m$ with fixed $m>0$.

\begin{proposition}[Closedness; \tier{Proved}]\label{prop:closed}
For every $m>0$, $\Ksurf_m$ is closed and weakly sequentially closed in $\Hs$.
\end{proposition}

\begin{proof}
Weak sequential closedness implies norm closedness, so we prove the former. Let $w_n\in\Ksurf_m$ with $w_n\rightharpoonup w$ in $\Hs$. By Proposition~\ref{prop:embedding}(i) the embeddings $H^2\hookrightarrow\hookrightarrow C^0(\overline{\Dom})$ and $H^2\hookrightarrow\hookrightarrow H^1$ are compact, so along the full sequence $w_n\to w$ in $C^0(\overline{\Dom})$ and $\nabla w_n\to\nabla w$ in $L^q(\Dom)$ for every $q<\infty$ (Rellich--Kondrachov applied to $\nabla w_n$ bounded in $H^1$), while $\partial^\beta w_n\rightharpoonup\partial^\beta w$ weakly in $L^2$ for $|\beta|=2$.

\emph{Floor.} $w_n\ge m$ and uniform convergence give $w\ge m$.

\emph{Calendar.} $\partial_\tau w_n\rightharpoonup\partial_\tau w$ weakly in $L^2$; the cone $\{u\in L^2: u\ge 0 \text{ a.e.}\}$ is convex and norm-closed, hence weakly closed; so $\partial_\tau w\ge 0$ a.e.

\emph{Butterfly.} Decompose $g[w_n]=\ell_n+q_n$ with $\ell_n=\tfrac12\partial^2_k w_n$ and $q_n$ the lower-order part of \eqref{eq:gfun}. We claim $q_n\to q$ strongly in $L^2$: indeed $1/w_n\to 1/w$ uniformly (denominators $\ge m$), $w_{n,k}\to w_k$ in $L^4$, hence $w_{n,k}^2\to w_k^2$ in $L^2$ and $k\,w_{n,k}/(2w_n)\to k\,w_k/(2w)$ in $L^4$, so the square converges in $L^2$; products of an $L^2$-convergent factor with a uniformly convergent bounded factor converge in $L^2$. Meanwhile $\ell_n\rightharpoonup\ell$ weakly in $L^2$. Therefore $g[w_n]\rightharpoonup g[w]$ weakly in $L^2$, and the weakly closed cone $\{u\ge 0\}$ again yields $g[w]\ge 0$ a.e.
\end{proof}

Weak closedness is what makes $\Ksurf_m$ usable: bounded minimizing sequences for projection and calibration problems have weak limit points that remain admissible, so metric projections onto $\Ksurf_m$ exist (uniqueness fails in general precisely because of the nonconvexity established next).

\subsection{Convexity in price coordinates, nonconvexity in variance coordinates}\label{sec:convexity}

Map the constraints through \eqref{eq:bs}. Writing $c(k,\tau)=B(k,w(k,\tau))$ and using $\partial_K = K^{-1}\partial_k$ for actual strike $K=Fe^{k}$, the classical characterization of arbitrage-free normalized call prices on a compact window \citep{carr2005note,davis2007range,roper2010arbitrage} reads: $(1-e^k)_+\le c\le 1$; $\;\partial_\tau c\ge 0$ a.e.; and convexity in strike, which in log-moneyness coordinates is the linear differential inequality
\begin{equation}\label{eq:convexk}
(\partial_{kk} - \partial_k)\, c \;\ge\; 0 \qquad \text{a.e.}
\end{equation}
Define $\Kc$ as the set of $c\in\Hs$ satisfying these conditions.

For $0<m<\Mcap<\infty$, define the bounded variance band and its price image
\begin{equation}\label{eq:bandedcharts}
 \Ksurf_{m,\Mcap}\coloneqq\{w\in\Ksurf_m:w\le \Mcap\},\qquad
 \Kcmbar\coloneqq\Kc\cap
 \{B(k,m)\le c(k,\tau)\le B(k,\Mcap)\text{ on }\overline\Dom\}.
\end{equation}
The upper band is not cosmetic: the lower price floor excludes the intrinsic
boundary, while $c\le B(k,\Mcap)<1$ excludes the upper wall where Black--Scholes
vega vanishes and $B^{-1}$ ceases to be uniformly regular.

\begin{proposition}[Exact correspondence of the variance and price charts;
\tier{Proved}]
\label{prop:chartcorrespondence}
For every $0<m<\Mcap<\infty$, the pointwise Black--Scholes map is a smooth
bijection
\[
 B:\Ksurf_{m,\Mcap}\longrightarrow\Kcmbar,
 \qquad B^{-1}:\Kcmbar\longrightarrow\Ksurf_{m,\Mcap},
\]
and both maps are restrictions of smooth Nemytskii maps on $H^2$-neighborhoods
of the displayed bands. Moreover, for $c=B[w]$,
\begin{equation}\label{eq:chartidentities}
 \partial_\tau c=B_w(k,w)\,\partial_\tau w,
 \qquad
 (\partial_{kk}-\partial_k)c=\frac{n(d_+)}{\sqrt w}\,g[w].
\end{equation}
Thus calendar and butterfly inequalities are equivalent in the two charts.
The set $\Kcmbar$ is closed, convex, and weakly closed in $\Hs$.
\end{proposition}

\begin{proof}
On the compact set
$[k_{\min},k_{\max}]\times[m,\Mcap]$, $B$ and its pointwise inverse are smooth,
$B_w$ is bounded above and away from zero, and all derivatives needed by the
$H^2$ chain rule are bounded. More explicitly,
\[
 c_k=B_k+B_ww_k,
 \qquad
 c_{kk}=B_{kk}+2B_{kw}w_k+B_{ww}w_k^2+B_ww_{kk},
\]
with analogous weak derivative formulae in $\tau$. The only quadratic term,
$B_{ww}w_k^2$, lies in $L^2$ because $w_k\in H^1\hookrightarrow L^4$.
The inverse direction has the same form, since the derivatives of $B^{-1}$
are bounded on the compact price band. Thus the associated Nemytskii maps
preserve $H^2$ and are smooth on neighborhoods of the bands. Strict
monotonicity in \eqref{eq:vega} gives the
pointwise bijection and the order bounds in \eqref{eq:bandedcharts}.
Differentiation in $\tau$ gives the first identity in
\eqref{eq:chartidentities}; the density calculation behind
\eqref{eq:density} gives the second. Both multipliers are strictly positive,
so the constraint signs are equivalent. Finally, $\Kcmbar$ is the intersection
of $\Kc$ with an order interval between two fixed $H^2$ functions. Each set is
closed and convex, hence so is the intersection; norm-closed convex subsets of
a Hilbert space are weakly closed.
\end{proof}

\begin{proposition}[Coordinate asymmetry; \tier{Proved}]\label{prop:convexity}
\begin{enumerate}[label=(\roman*)]
\item $\Kc$ is a closed convex subset of $\Hs$, and it is weakly closed.
\item On $\Dom=[-2,2]\times[\tau_{\min},\tau_{\max}]$ and for $m=10^{-3}$, the total-variance set $\Ksurf_m$ is not convex.
\end{enumerate}
\end{proposition}

\begin{proof}
(i) Each defining condition of $\Kc$ is an affine inequality composed with a bounded linear map on $\Hs$ ($c\mapsto c$ pointwise, $c\mapsto\partial_\tau c$, $c\mapsto(\partial_{kk}-\partial_k)c$, each into $C^0$ or $L^2$), followed by intersection with a closed convex cone or order interval. Intersections of closed convex sets are closed convex; convex plus norm-closed implies weakly closed.

(ii) The map $B$ of \eqref{eq:bs} is strictly monotone but nonlinear in $w$, with second derivative
\begin{equation}\label{eq:vomma}
\partial^2_w B(k,w) \;=\; \frac{n(d_+)}{2\sqrt{w}}\;\Bigl(\frac{d_+ d_-}{2w} \;-\; \frac{1}{2w}\Bigr)
\;=\; \frac{n(d_+)}{4 w^{3/2}}\,\bigl(d_+ d_- - 1\bigr),
\end{equation}
whose sign changes across $\Dom$. Here is an explicit certificate. Let
\[
w_i(k)=a_i+b_i\{\rho_i(k-\mu_i)+[(k-\mu_i)^2+s_i^2]^{1/2}\},
\]
independent of $\tau$, with
\begin{align*}
(a_1,b_1,\rho_1,\mu_1,s_1)&=(1.18496492,0.76812442,0.57526503,1.00347025,0.1311388577407955),\\
(a_2,b_2,\rho_2,\mu_2,s_2)&=(-0.06074514,1.80133337,0.07870247,1.49987078,0.034545633915444976).
\end{align*}
Directed-rounding interval arithmetic on $[-2,2]$ gives
\[
\inf w_1\ge1.2672888,\quad \inf g[w_1]\ge0.0037007,
\qquad
\inf w_2\ge0.0012673,\quad \inf g[w_2]\ge0.0029721.
\]
Thus both slices lie in $\Ksurf_{10^{-3}}$ (their calendar derivative is zero),
whereas for $\bar w=(w_1+w_2)/2$,
\[
g[\bar w](1.75010)\le-0.5871470.
\]
Hence $\bar w\notin\Ksurf_{10^{-3}}$. The standard-library verifier
\texttt{verify\_svi\_nonconvexity.py} ships with the source and encloses every
operation using 50-digit decimal arithmetic with outward rounding on 25,000
subintervals; the displayed margins are interval bounds, not sampled minima.
\end{proof}

For every $\Mcap\ge\max_{i=1,2}\sup_{[-2,2]}w_i$, the same certificate lies
inside the bounded band and proves that $\Ksurf_{10^{-3},\Mcap}$ is nonconvex.
Proposition~\ref{prop:chartcorrespondence} therefore gives a smooth
homeomorphism from this nonconvex variance set onto the convex price set
$\Kcmbar$: the chart, rather than a relaxation, absorbs the nonconvexity.

The asymmetry of Proposition~\ref{prop:convexity} organizes the rest of the paper: whenever convexity is needed (projections with uniqueness, reflected dynamics, feasibility of quadratic programs), we pass to price coordinates; whenever spectral and hedging statements are needed in the coordinates the desk quotes, we work in $w$ and keep the geometry nonlinear.

\begin{figure}[htbp]
\centering
\includegraphics[width=0.98\textwidth]{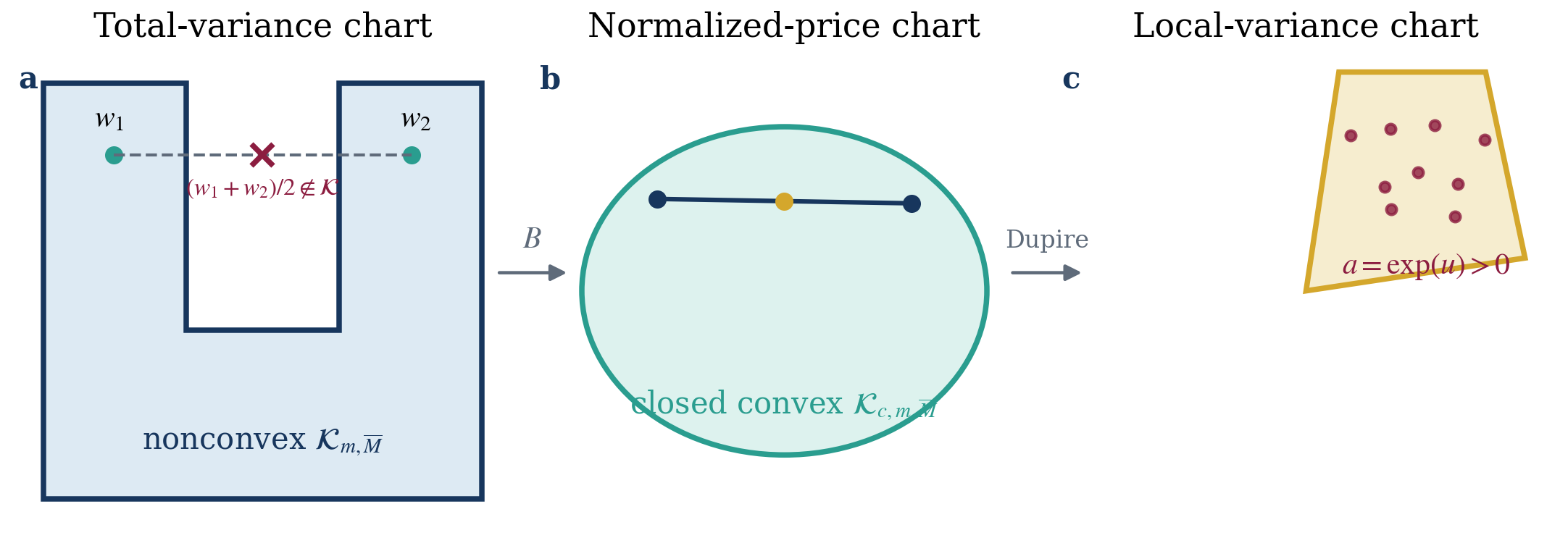}
\caption{Schematic geometry of the three surface charts. The total-variance
chart can be nonconvex (left), while the normalized-price chart is convex on a
fixed variance band (center) and the local-variance chart is the positive cone
(right). The left panel illustrates the midpoint obstruction; the rigorous
numerical certificate is Proposition~\ref{prop:convexity}(ii). No market data
are plotted.}
\label{fig:threecharts}
\end{figure}

\section{First-order geometry: tangent cones and feasible directions}\label{sec:geometry}

\subsection{Active sets and the linearization cone}\label{sec:active}

Fix $w\in\Ksurf_m$. Define the active sets
\begin{equation}\label{eq:active}
A_{\mathrm{fl}}(w) = \{x\in\overline{\Dom}: w(x)=m\},\qquad
A_{\mathrm{cal}}(w) = \{\partial_\tau w = 0\},\qquad
A_{\mathrm{bf}}(w) = \{g[w]=0\},
\end{equation}
the latter two up to null sets. The \emph{contingent (Bouligand) cone} to $\Ksurf_m$ at $w$ is
\begin{equation}\label{eq:contingent}
T_{\Ksurf_m}(w) \;=\; \Bigl\{ h\in\Hs:\ \exists\, t_n\downarrow 0,\ h_n\to h \text{ in } \Hs \text{ with } w+t_n h_n\in\Ksurf_m \Bigr\},
\end{equation}
and the \emph{linearization cone} is
\begin{equation}\label{eq:lincone}
L(w) \;=\; \Bigl\{ h\in\Hs:\ h\ge 0 \text{ on } A_{\mathrm{fl}}(w),\ \ \partial_\tau h\ge 0 \text{ a.e.\ on } A_{\mathrm{cal}}(w),\ \ Dg[w]h\ge 0 \text{ a.e.\ on } A_{\mathrm{bf}}(w) \Bigr\},
\end{equation}
where $Dg[w]$ is the directional derivative of the butterfly functional computed next.

For localization of almost-everywhere constraints, a sharper closed set is
useful. If $0\le u\in L^2(\Dom)$, define its \emph{essential contact set}
\begin{equation}\label{eq:essentialcontact}
 Z_{\mathrm{ess}}(u)
 \coloneqq
 \left\{x\in\overline\Dom:
 \operatorname*{ess\,inf}_{B_r(x)\cap\Dom}u=0
 \text{ for every }r>0\right\}.
\end{equation}
This is closed, and if an open set $U$ contains $Z_{\mathrm{ess}}(u)$, then
compactness gives a $\gamma>0$ such that $u\ge\gamma$ a.e. on
$\Dom\setminus U$: each point of the complement has a ball with positive
essential infimum, and finitely many such balls suffice. For a continuous
nonnegative representative, $Z_{\mathrm{ess}}(u)=\{u=0\}$. We use
\eqref{eq:essentialcontact} only to localize uniform margins; it does not
replace the $L^2$-cone active set in \eqref{eq:lincone}, since a field may be
positive a.e. while having zero local essential infimum.

\begin{lemma}[Directional derivative of $g$; \tier{Proved}]\label{lem:Dg}
Let $w\in\Hs$ with $w\ge m>0$ and let $A=1-\dfrac{k\,w_k}{2w}$. For $h\in\Hs$,
\begin{equation}\label{eq:Dg}
Dg[w]h \;=\; -\,\frac{k\,A}{w^{2}}\,\bigl(w\,h_k - w_k\,h\bigr)
\;-\; \frac{w_k h_k}{2}\Bigl(\frac1w+\frac14\Bigr)
\;+\; \frac{w_k^{2}}{4w^{2}}\,h
\;+\; \frac{h_{kk}}{2},
\end{equation}
and $g:\{w\ge m\}\subset\Hs\to L^2(\Dom)$ is Fr\'echet differentiable with derivative \eqref{eq:Dg}, locally Lipschitz on bounded subsets of $\{w\ge m\}$.
\end{lemma}

The computation, together with the second-order remainder estimate that makes the Fr\'echet claim precise, is Appendix~\ref{app:Dg}. Note the structure of \eqref{eq:Dg}: the only term that sees the full $H^2$ topology is the linear term $h_{kk}/2$; all lower-order terms are continuous already in $H^1\cap C^0$. This split --- weak top order, strong lower order --- is the same one that drove the closedness proof, and it will drive the boundary theorem.

\subsection{First-order feasibility}\label{sec:feasibility}

\begin{proposition}[Feasible directions; \tier{Proved}]\label{prop:feasible}
Let $w\in\Ksurf_m\cap C^2(\overline\Dom)$ and let
$h\in\Hs\cap C^2(\overline{\Dom})$ satisfy, for some $\varepsilon>0$ and open
neighborhoods $U_{\mathrm{fl}},U_{\mathrm{cal}},U_{\mathrm{bf}}$ of
$A_{\mathrm{fl}}(w)$,
$Z_{\mathrm{ess}}(\partial_\tau w)$, and $Z_{\mathrm{ess}}(g[w])$,
respectively,
\[
h\ge\varepsilon \text{ on } U_{\mathrm{fl}},\qquad
\partial_\tau h\ge\varepsilon \text{ a.e.\ on } U_{\mathrm{cal}},\qquad
Dg[w]h\ge\varepsilon \text{ a.e.\ on } U_{\mathrm{bf}}.
\]
Then there is $t_0>0$ with $w+th\in\Ksurf_m$ for all $t\in[0,t_0]$. Consequently every such $h$ lies in $T_{\Ksurf_m}(w)$, and $T_{\Ksurf_m}(w)\subseteq L(w)$ always. \tier{Proved}
\end{proposition}

\begin{proof}
For smooth $w$, the continuous representatives make the essential contact
sets equal to their pointwise zero sets. The localization property following
\eqref{eq:essentialcontact} therefore gives a common $\gamma>0$ with
$w\ge m+\gamma$ off $U_{\mathrm{fl}}$,
$\partial_\tau w\ge\gamma$ off $U_{\mathrm{cal}}$, and
$g[w]\ge\gamma$ off $U_{\mathrm{bf}}$. The essential contact sets create the
inactive-region margin; the smooth-stratum assumption is still needed because
the nonlinear butterfly remainder must be controlled pointwise rather than
only in $L^1$ or $L^2$.
For the floor and calendar constraints, $w+th\ge m$ on $U_{\mathrm{fl}}$ for
$t\ge0$ and $\ge m+\gamma-t\norm{h}_{C^0}$ elsewhere; similarly for
$\partial_\tau(w+th)$. For the butterfly, Lemma~\ref{lem:Dg}
(Appendix~\ref{app:Dg}, estimate \eqref{eq:remainder}) gives
\[
g[w+th] \;=\; g[w] + t\,Dg[w]h + r_t,\qquad \norm{r_t}_{L^\infty} \le C_h\,t^2,
\]
with $C_h<\infty$ because $h\in C^2$ and $w+th\ge m$ for small $t$; on $U_{\mathrm{bf}}$ the right side is $\ge t\varepsilon - C_h t^2$ a.e., and off it $\ge \gamma - t\norm{Dg[w]h}_{L^\infty(\Dom\setminus U_{\mathrm{bf}})} - C_h t^2$, where the middle norm is finite again by $h\in C^2$ and Proposition~\ref{prop:embedding}. Choosing $t_0$ small makes all three families nonnegative. The inclusion $T_{\Ksurf_m}(w)\subseteq L(w)$ is the standard first-order necessary condition: if $w+t_nh_n\in\Ksurf_m$, pass to the limit in each constraint after dividing by $t_n$, using the differentiability of Lemma~\ref{lem:Dg} and the weak closedness of the nonnegative cones as in Proposition~\ref{prop:closed}.
\end{proof}

\begin{remark}[Constraint qualification; \tier{Conditional}]\label{rem:cq}
Equality $T_{\Ksurf_m}(w)=\overline{L(w)}$ holds under a Robinson-type constraint qualification: the existence of a single strictly feasible direction $\bar h$ as in Proposition~\ref{prop:feasible} \citep[Ch.~4]{aubin1990setvalued}. In price coordinates such a Slater direction exists at every $c\in\Kc$ with room above the intrinsic bound (push $c$ toward a fixed strictly arbitrage-free interior surface); pulling it back through $B^{-1}$ produces a candidate $\bar h$ in variance coordinates. We tier the equality \tier{Conditional} on the verification that the pullback direction retains strict margins for the nonlinear functional $Dg[w]$ at every boundary point of interest; the verification is mechanical at any given $w$ and is part of the projected-innovation implementation of Section~\ref{sec:constructions}.
\end{remark}

\section{The local-volatility chart: Dupire as constraint geometry}\label{sec:localchart}

The implied-variance and price charts have a third companion. Define the local
variance field $a=\sigma_{\mathrm{loc}}^2$ and the positive cone
$\Ka=\{a\in C(\overline\Dom):a\ge0\}$. The following identity makes this chart
intrinsic to the constraint geometry rather than an auxiliary pricing device.

\begin{theorem}[Dupire in constraint coordinates; \tier{Proved}]
\label{thm:dupireconstraint}
Let $w\in H^s_\rho(\Dom)$, $s\ge3$, with $w,w_k,w_{kk},w_\tau$ continuous on
$\overline\Dom$ and with uniform margins $w\ge m>0$ and $g[w]\ge\gamma>0$.
Let $c(k,\tau)=B(k,w(k,\tau))$ be its normalized call-price surface. In the
forward-normalized, zero-carry convention \citep{dupire1994pricing}, Dupire
local variance satisfies
\begin{equation}\label{eq:dupireratio}
 \boxed{\quad a(k,\tau)=\frac{\partial_\tau w(k,\tau)}{g[w](k,\tau)}\quad}.
\end{equation}
Equivalently, the price chart solves the forward equation
\begin{equation}\label{eq:dupireforward}
 \partial_\tau c=\frac12a(\partial_{kk}-\partial_k)c.
\end{equation}
\end{theorem}

\begin{proof}
The total-variance vega is
$\partial_wB=n(d_+)/(2\sqrt w)$, hence
$\partial_\tau c=n(d_+)\partial_\tau w/(2\sqrt w)$. The
Gatheral--Jacquier density identity gives
$(\partial_{kk}-\partial_k)c=n(d_+)g[w]/\sqrt w$. Dividing the two identities
proves \eqref{eq:dupireratio}; substitution yields \eqref{eq:dupireforward}.
\end{proof}

\begin{corollary}[Zeros and poles; \tier{Proved}]\label{cor:dupiredivisor}
On an admissible surface, the zero set of $a$ inside $\{g>0\}$ is exactly the
calendar-active set $\{\partial_\tau w=0,g>0\}$. The pole set
$\{g=0,\partial_\tau w>0\}$ lies in the butterfly-active stratum. On the doubly
active set $\{g=\partial_\tau w=0\}$ the ratio is indeterminate. Positivity of
$a$ alone is not a certificate of static arbitrage: simultaneous negativity of
numerator and denominator loses one sign bit, which convexity in the price
chart restores through \eqref{eq:dupireforward}.
\end{corollary}

The three charts therefore separate three kinds of simplicity:
\begin{center}
\begin{tabularx}{\textwidth}{@{}lllX@{}}
\toprule
Chart & State & Geometry & Cost of the chart map\\
\midrule
Implied & $w$ & closed, generally nonconvex & native market quotation\\
Price & $c$ & closed convex set & pointwise Black--Scholes map\\
Local & $a$ & closed convex cone & differential ratio / parabolic inverse\\
\bottomrule
\end{tabularx}
\end{center}
On the strict interior, $a\mapsto\log a$ is a homeomorphism from $\{a>0\}$
onto $C(\overline\Dom)$. Thus positivity can be coordinatized away, but only by
placing the boundary at infinite log-distance. The chart removes a static
constraint; it does not remove the dynamic consistency restriction linking the
evolution of $a$ and the option-price codebook.

\begin{proposition}[Stable inverse local chart; \tier{Proved, cited}]
\label{prop:localinverse}
Let $I=(k_-,k_+)$, $J=[\tau_0,\tau_1]$, and $Q=I\times J$. Fix
$c_0\in H^2(I)$ and Dirichlet traces $b_\pm\in H^1(J)$ satisfying the corner
compatibility conditions. Let $a_i\in W^{1,\infty}(Q)$, $i=1,2$, obey
\[
0<\underline a\le a_i\le\overline a<\infty,
\qquad \|a_i\|_{W^{1,\infty}(Q)}\le A.
\]
Then \eqref{eq:dupireforward}, with $c(\cdot,\tau_0)=c_0$ and
$c(k_\pm,\cdot)=b_\pm$, has a unique weak solution
$c[a_i]\in C(J;L^2(I))\cap L^2(J;H^2(I))$; it is classical under the usual
parabolic H\"older compatibility assumptions. Moreover,
\begin{align}\label{eq:localstability}
 &\sup_{\tau\in J}\|c[a_1](\cdot,\tau)-c[a_2](\cdot,\tau)\|_{L^2(I)}^2
 +\int_{\tau_0}^{\tau_1}\!\|\partial_k(c[a_1]-c[a_2])\|_{L^2(I)}^2\dd\tau
 \nonumber\\[-2pt]
 &\hspace{4cm}\le K\|a_1-a_2\|_{L^\infty(Q)}^2,
\end{align}
where $K$ depends only on $I,J,\underline a,\overline a,A$ and the displayed
data norms. If the data are restrictions of the whole-line Dupire solution
with the natural call-price wing limits, then $c[a_i]$ preserves the static
call-price inequalities; on a truncated interval this preservation requires
boundary traces compatible with that global solution.
\end{proposition}

\begin{proof}[Proof sketch]
Uniform parabolicity gives existence and uniqueness. Subtracting the two
equations yields a zero-data parabolic equation whose forcing is
$\tfrac12(a_1-a_2)(\partial_{kk}-\partial_k)c[a_2]$. The standard $L^2$
parabolic estimate, followed by Young's inequality and Gronwall, gives
\eqref{eq:localstability}; maximal regularity bounds the forcing norm by the
stated data and ellipticity constants. On the whole line, the representation
as call prices under
$\dd X_\tau=-\tfrac12a(X_\tau,\tau)\dd\tau+\sqrt{a(X_\tau,\tau)}\dd W_\tau$
preserves the option bounds, strike convexity, and maturity monotonicity.
Restriction proves the last assertion. See \citet{friedman1964parabolic} for
the parabolic estimates and regularity used here.
\end{proof}

\begin{proposition}[SSVI local variance; \tier{Proved}]\label{prop:ssvilocal}
For SSVI, $w=\frac\theta2(1+\rho\phi k+R)$ with
$R=\sqrt{(\phi k+\rho)^2+1-\rho^2}$, so
\begin{equation}\label{eq:ssviderivatives}
w_k=\frac{\theta\phi}{2}\left(\rho+\frac{\phi k+\rho}{R}\right),
\qquad
w_{kk}=\frac{\theta\phi^2(1-\rho^2)}{2R^3},
\end{equation}
and, when $\phi=\phi(\theta)$,
\begin{equation}\label{eq:ssvitur}
w_\tau=\frac{\theta'_\tau}{2}
\left[1+\rho\phi k+R
+\theta\phi'(\theta)k\left(\rho+\frac{\phi k+\rho}{R}\right)\right].
\end{equation}
Equations~\eqref{eq:ssviderivatives}--\eqref{eq:ssvitur} inserted into
\eqref{eq:dupireratio} give the full closed form. At
the money,
\begin{equation}\label{eq:ssviatm}
a(0,\tau)=\frac{\theta'_\tau}
{1+\frac{\theta\phi^2}{4}(1-2\rho^2)-\frac{(\theta\phi\rho)^2}{16}}.
\end{equation}
This makes the SSVI butterfly conditions readable as positivity and boundedness
of the local ellipticity denominator on the quoted window.
\end{proposition}

\begin{proof}
Since $R_k=\phi(\phi k+\rho)/R$ and
$\partial_k[(\phi k+\rho)/R]=\phi(1-\rho^2)/R^3$, direct differentiation gives
\eqref{eq:ssviderivatives}. The maturity derivative is obtained by the chain
rule through $\theta(\tau)$ and $\phi(\theta)$; inserting it and
\eqref{eq:ssviderivatives} into \eqref{eq:gfun} gives the full elementary ratio
$a=w_\tau/g[w]$. At $k=0$, $R=1$, $w=\theta$,
$w_k=\theta\phi\rho$, $w_{kk}=\theta\phi^2(1-\rho^2)/2$, and
$w_\tau=\theta'_\tau$. Substitution into $g[w]$ yields
\[
g[w](0,\tau)=1+\frac{\theta\phi^2}{4}(1-2\rho^2)
-\frac{(\theta\phi\rho)^2}{16},
\]
which proves \eqref{eq:ssviatm}.
\end{proof}

\section{Static feasibility versus dynamic no-arbitrage}
\label{sec:dynamicarb}

Membership in $\Ksurf_m$ excludes butterfly and calendar arbitrage at one
observation time. It does not by itself make the evolution of option prices
arbitrage-free. The missing object is the roll of absolute expiry through the
time-to-maturity chart, followed by the martingale restriction for each fixed
contract. This distinction is the option-surface analogue of the separation
between positivity of an HJM curve and its risk-neutral drift restriction
\citep{carmona2009local,kallsen2015hjm}.

\begin{proposition}[Musiela transport and fixed-contract restriction;
\tier{Proved}]
\label{prop:musiela}
Let $\widehat w_t(k,T)$ be an absolute-expiry total-variance codebook satisfying,
for each fixed $(k,T)$,
\begin{equation}\label{eq:absoluteexpiry}
 \dd\widehat w_t(k,T)=b_t(k,T)\dd t
 +\sum_{r=1}^d\sigma_t^r(k,T)\dd B_t^r,
\end{equation}
with $C^1$ regularity in $T$ and $C^2$ regularity in $k$. Define the Musiela
field $w_t(k,\tau)=\widehat w_t(k,t+\tau)$ and the shifted coefficients
$b_t^M(k,\tau)=b_t(k,t+\tau)$ and
$\sigma_t^{M,r}(k,\tau)=\sigma_t^r(k,t+\tau)$. Assume also that $b$ and each
$\sigma^r$ are jointly continuous in their arguments and $C^1$ in $k$, with
the local bounds required by the It\^o--Wentzell formula. Then
\begin{equation}\label{eq:musiela}
 \dd w_t(k,\tau)=
 \bigl(\partial_\tau w_t(k,\tau)+b_t^M(k,\tau)\bigr)\dd t
 +\sum_{r=1}^d\sigma_t^{M,r}(k,\tau)\dd B_t^r.
\end{equation}
For a fixed contract $(K,T)$, put $\tau_t=T-t$ and
$k_t=\log(K/F_t(T))$, and suppose
$\dd k_t=\beta_t^k\dd t+\sum_r\gamma_t^r\dd B_t^r$. Then the
$\partial_\tau w$ roll term cancels along the contract and It\^o--Wentzell gives
\begin{align}\label{eq:fixedcontract}
 \dd w_t(k_t,\tau_t)
 &=\Bigl[b_t^M+\beta_t^k w_k
 +\frac12\sum_r(\gamma_t^r)^2w_{kk}
 +\sum_r\gamma_t^r\partial_k\sigma_t^{M,r}\Bigr]_{(k_t,\tau_t)}\dd t
 \nonumber\\
 &\quad+\sum_r\bigl[\sigma_t^{M,r}+\gamma_t^r w_k\bigr]_{(k_t,\tau_t)}\dd B_t^r.
\end{align}
Consequently, $w_t\in\Ksurf_m$ for every $t$ is a static invariance condition.
Dynamic absence of arbitrage additionally requires the discounted,
non-normalized call price of every fixed $(K,T)$ to be a local martingale under
an equivalent pricing measure; applying It\^o--Wentzell to the Black--Scholes
map $B(k_t,w_t(k_t,\tau_t))$, together with the forward and discount factors,
imposes a separate drift restriction on $(b^M,\sigma^M,F,D)$.
\end{proposition}

\begin{proof}
The deterministic chain rule for $T=t+\tau$ applied to
\eqref{eq:absoluteexpiry} gives \eqref{eq:musiela}. Compose that random field
with $(k_t,\tau_t)$, use $\dd\tau_t=-\dd t$, and apply the It\^o--Wentzell
formula. The $+\partial_\tau w_t\dd t$ term in \eqref{eq:musiela} cancels
$w_\tau\dd\tau_t$; the remaining drift, quadratic-variation, and cross-variation
terms are exactly \eqref{eq:fixedcontract}; see \citet{kunita1990stochastic}
for the stochastic field composition formula. The final assertion is the
fundamental discounted-price martingale condition under an equivalent pricing
measure. The option-surface forms of the resulting drift restriction are
developed by \citet{carmona2009local,kallsen2015hjm}.
\end{proof}

\begin{remark}[What the constructions guarantee]
\label{rem:staticdynamic}
The heads, projections, reflections, and flow maps below guarantee
$w_t\in\Ksurf_m$ or $c_t\in\Kc$ at every generated time. They are therefore
exactly statically arbitrage-free. They become dynamically arbitrage-free market
models only after their drift and joint forward/discount dynamics satisfy the
fixed-contract condition of Proposition~\ref{prop:musiela}. The empirical
study treats them as physical-measure scenario generators unless that
additional restriction is explicitly imposed.
\end{remark}

\section{Dynamics: the boundary incompatibility theorem}\label{sec:dynamics}

\subsection{Gaussian field dynamics}\label{sec:gaussdyn}

The natural infinite-dimensional lift of surface PCA is the $\Hs$-valued It\^o equation
\begin{equation}\label{eq:sde}
\dd w_t \;=\; \mu_t\,\dd t \;+\; \Sigma\,\dd W_t,
\end{equation}
with $W$ a cylindrical Wiener process on $\Hs$, $\Sigma$ Hilbert--Schmidt so that $C=\Sigma\Sigma^\ast$ is a trace-class, self-adjoint, nonnegative covariance operator, and $\mu$ adapted with locally bounded norms \citep{daprato2014stochastic}. The one-step Euler shock from $w_0$ over $[0,\Delta]$ is
\begin{equation}\label{eq:euler}
w_\Delta \;=\; w_0 \;+\; \Delta\,\mu \;+\; \sqrt{\Delta}\,\xi, \qquad \xi\sim N(0,C) \text{ on } \Hs .
\end{equation}
The question is whether \eqref{eq:sde}--\eqref{eq:euler} can respect $\Ksurf_m$. The answer splits cleanly by the location of $w_0$.

\subsection{Smoothed constraint functionals}\label{sec:smoothed}

Pointwise evaluation of $g[w]$ is not continuous on $\Hs$ (Proposition~\ref{prop:embedding}(iii)), so boundary activity is detected through localized averages. For $0\le\psi\in L^\infty(\Dom)$, $\psi\not\equiv0$, define
\begin{equation}\label{eq:Phi}
\Phi_\psi^{\mathrm{bf}}(w) = \int_\Dom g[w]\,\psi \dd x,\qquad
\Phi_\psi^{\mathrm{cal}}(w) = \int_\Dom \partial_\tau w\,\psi \dd x,\qquad
\Phi_\psi^{\mathrm{fl}}(w) = \int_\Dom (w-m)\,\psi \dd x .
\end{equation}
Each is finite on $\{w\ge m\}$, nonnegative on $\Ksurf_m$, and vanishes exactly when the corresponding constraint is active a.e.\ on $\{\psi>0\}$. By Lemma~\ref{lem:Dg}, $\Phi_\psi^{\mathrm{bf}}$ is $C^1$ on $\{w\ge m\}\subset\Hs$ with locally Lipschitz derivative
\begin{equation}\label{eq:DPhi}
D\Phi_\psi^{\mathrm{bf}}[w]h \;=\; \int_\Dom \bigl(Dg[w]h\bigr)\,\psi \dd x
\;=\; \ip{\ell_\psi(w)}{h}_{\Hs},
\end{equation}
the last identity by Riesz representation; the calendar and floor functionals are bounded linear, hence trivially $C^1$. Call $\psi$ an \emph{active window} at $w_0$ for a constraint family if the corresponding $\Phi_\psi(w_0)=0$.

\begin{lemma}[Second differentiability of the butterfly window;
\tier{Proved}]
\label{lem:PhiC2}
Let $0\le\psi\in L^\infty(\Dom)$. On the open set
$\mathcal O=\{w\in\Hs:\inf_{\overline\Dom}w>0\}$, the scalar functional
$\Phi_\psi^{\mathrm{bf}}$ is twice continuously Fr\'echet differentiable.
Writing $A=1-kw_k/(2w)$, the diagonal second variation of the underlying
butterfly map is
\begin{align}\label{eq:D2g}
D^2g[w][h,h]
={}&\frac{k^2}{2w^4}(wh_k-w_kh)^2
 +\frac{2Ak}{w^2}\left(hh_k-\frac{w_k}{w}h^2\right) \notag\\
&-\frac{h_k^2}{2}\left(\frac1w+\frac14\right)
 +\frac{w_k}{w^2}hh_k-\frac{w_k^2}{2w^3}h^2,
\end{align}
and polarization gives the symmetric bilinear map
$D^2g[w]:\Hs\times\Hs\to L^2(\Dom)$. Consequently
\[
 D^2\Phi_\psi^{\mathrm{bf}}[w](h_1,h_2)
 =\int_\Dom D^2g[w][h_1,h_2]\,\psi\dd x.
\]
For
every $m,R>0$ there is $K=K(m,R,\psi)$ such that
\begin{equation}\label{eq:PhiHessian}
 \bigl|D^2\Phi_\psi^{\mathrm{bf}}[w](h_1,h_2)\bigr|
 \le K\norm{h_1}_{\Hs}\norm{h_2}_{\Hs}
\end{equation}
whenever $w\ge m$ and $\norm{w}_{\Hs}\le R$. The Hessian is locally
Lipschitz in $w$ in operator norm on bounded subsets of this set.
\end{lemma}

\begin{proof}
The top-order term $w_{kk}/2$ in $g[w]$ is linear, so it contributes no
Hessian. Differentiating the remaining terms twice gives \eqref{eq:D2g}; no
term contains $h_{kk}$, and polarization supplies the cross-variation. In two
dimensions $H^2\hookrightarrow C^0$ and $H^1\hookrightarrow L^4$. Thus
$wh_k-w_kh\in L^4$, $A\in L^4$, and every product in \eqref{eq:D2g} belongs
to $L^2$, with norm bounded by a constant times $\norm{h}_{\Hs}^2$.
Polarization and integration against $\psi$ give \eqref{eq:PhiHessian}.
Applying the same estimates to coefficient differences, using the local
Lipschitz property of $u\mapsto u^{-j}$ on $[m,\infty)$ and
$w_k\in L^4$, proves the stated operator-norm Lipschitz property. This also
promotes the quadratic estimate in Appendix~\ref{app:Dg} from a remainder
statement to the $C^2$ hypothesis required by the Hilbert-space It\^o formula.
\end{proof}

\begin{remark}[Window activity and isolated contacts]
\label{rem:pointcontact}
Theorem~\ref{thm:boundary} is deliberately an $H^2$ result. For the calendar
and butterfly constraints, an active window means equality on a set of positive
measure inside $\{\psi>0\}$; an isolated point contact is invisible to the
$L^2$ order structure. If the state is strengthened to $H^s_\rho(\Dom)$ with
$s>3$, Sobolev embedding gives $C^2$ representatives. Then point evaluation of
$w_\tau$ and $g[w]$ is continuous, one may use
$\Phi_x^{\mathrm{cal}}(w)=w_\tau(x)$ or
$\Phi_x^{\mathrm{bf}}(w)=g[w](x)$, and the same proof yields the half-law with
the pointwise normal. The finite-grid theorem below is the corresponding
discrete treatment of isolated contacts.
\end{remark}

\subsection{The \texorpdfstring{$1/2$}{one-half} law at the boundary}\label{sec:halflaw}

\begin{theorem}[Boundary incompatibility; \tier{Proved}]\label{thm:boundary}
Let $w_0\in\Ksurf_m$, let $\Phi$ be any of the functionals \eqref{eq:Phi} with $\Phi(w_0)=0$, and suppose the noise is nondegenerate along the constraint normal:
\begin{equation}\label{eq:nondeg}
\sigma_\Phi^2 \;\coloneqq\; \ip{\ell}{C\,\ell}_{\Hs} \;>\; 0, \qquad \ell \coloneqq \text{Riesz representer of } D\Phi[w_0].
\end{equation}
Then for every fixed drift value $\mu\in\Hs$, the Euler shock \eqref{eq:euler} satisfies
\begin{equation}\label{eq:halflimit}
\lim_{\Delta\downarrow 0}\ \Prob\bigl[\,\Phi(w_\Delta) < 0\,\bigr] \;=\; \tfrac12 .
\end{equation}
In particular,
\begin{equation}\label{eq:violationliminf}
 \liminf_{\Delta\downarrow0}\Prob[w_\Delta\notin\Ksurf_m]\ge\tfrac12.
\end{equation}
\end{theorem}

\begin{proof}
Write $u_\Delta=\Delta\mu+\sqrt{\Delta}\,\xi$, so $\norm{u_\Delta}_{\Hs}=O_{\Prob}(\sqrt{\Delta})$ since $\E\norm{\xi}^2_{\Hs}=\tr C<\infty$. $C^1$ differentiability with locally Lipschitz derivative at $w_0$ gives
\[
\Phi(w_0+u_\Delta) \;=\; \underbrace{\Phi(w_0)}_{=0} \;+\; \ip{\ell}{u_\Delta}_{\Hs} \;+\; r(u_\Delta),
\qquad |r(u)| \le \kappa\,\norm{u}_{\Hs}^{2}
\]
for $\norm{u}_{\Hs}$ small (for the linear functionals $r\equiv 0$; for $\Phi^{\mathrm{bf}}_\psi$ this is Appendix~\ref{app:Dg}, estimate \eqref{eq:remainder} integrated against $\psi$). Divide by $\sqrt{\Delta}$:
\[
\frac{\Phi(w_\Delta)}{\sqrt{\Delta}} \;=\; \ip{\ell}{\xi}_{\Hs} \;+\; \sqrt{\Delta}\,\ip{\ell}{\mu}_{\Hs} \;+\; \frac{r(u_\Delta)}{\sqrt{\Delta}} .
\]
The first term is exactly $N(0,\sigma_\Phi^2)$ for every $\Delta$; the second is $O(\sqrt{\Delta})$; and
\[
\frac{|r(u_\Delta)|}{\sqrt{\Delta}}
\le \kappa\frac{\norm{u_\Delta}^2}{\sqrt{\Delta}}
=O_{\Prob}(\sqrt{\Delta})\longrightarrow0.
\]
By Slutsky, $\Phi(w_\Delta)/\sqrt{\Delta}\Rightarrow N(0,\sigma_\Phi^2)$, and since the limit law has no atom at $0$,
$\Prob[\Phi(w_\Delta)<0]\to\Prob[N(0,\sigma_\Phi^2)<0]=\tfrac12$. The final statement follows because $\{\Phi(w_\Delta)<0\}\subseteq\{w_\Delta\notin\Ksurf_m\}$.
\end{proof}

\begin{corollary}[The $\sqrt\Delta$ boundary layer; \tier{Proved}]
\label{cor:boundarylayer}
Let $\Phi$ be one of the window functionals in \eqref{eq:Phi}. Let
$\mathcal U\subset\mathcal O$ have a common $\Hs$-neighborhood on which
$\Phi$ is $C^2$, and assume there are constants $L<\infty$ and
$\sigma_0>0$ such that, throughout that neighborhood,
\[
 \norm{D\Phi[w]}_{\Hs^*}+\norm{D^2\Phi[w]}_{\mathrm{op}}\le L,
 \qquad
 \sigma_\Phi(w)\coloneqq
 \sqrt{\ip{\ell(w)}{C\ell(w)}_{\Hs}}\ge\sigma_0.
\]
Let $\mu:\mathcal U\to\Hs$ be bounded and set
$w_\Delta(w)=w+\Delta\mu(w)+\sqrt\Delta\,\xi$ with
$\xi\sim N(0,C)$. Writing $\overline N=1-N$, for every $R<\infty$,
\begin{equation}\label{eq:uniformboundarylayer}
 \sup_{\substack{w\in\mathcal U\cap\Ksurf_m\\
                  0\le\Phi(w)\le R\sqrt\Delta}}
 \left|
 \Prob\!\left[\Phi(w_\Delta(w))<0\right]
 -\overline N\!\left(
   \frac{\Phi(w)}{\sqrt\Delta\,\sigma_\Phi(w)}
  \right)
 \right|\longrightarrow0.
\end{equation}
In particular, if
$\Phi(w_0^\Delta)/(\sqrt\Delta\,\sigma_\Phi(w_0^\Delta))\to r\ge0$,
then the violation probability converges to $\overline N(r)$; $r=0$ recovers
Theorem~\ref{thm:boundary}.
\end{corollary}

\begin{proof}
Uniform Taylor expansion on the common neighborhood gives
\[
 \Phi(w_\Delta(w))
 =\Phi(w)+\sqrt\Delta\,\ip{\ell(w)}{\xi}_{\Hs}
 +\Delta\ip{\ell(w)}{\mu(w)}_{\Hs}+r_{\Delta,w},
\]
where
$|r_{\Delta,w}|\le(L/2)\norm{\Delta\mu(w)+\sqrt\Delta\xi}_{\Hs}^2$
whenever the step remains in that neighborhood. Since $C$ is trace class and
$\mu$ is bounded,
\[
 \sup_{w\in\mathcal U}
 \Prob\!\left[
 \frac{|\Delta\ip{\ell(w)}{\mu(w)}+r_{\Delta,w}|}
      {\sqrt\Delta\,\sigma_\Phi(w)}>\varepsilon
 \right]\longrightarrow0
\]
for every $\varepsilon>0$; the probability of leaving the common neighborhood
also tends to zero uniformly by Markov's inequality. The leading random term
divided by $\sigma_\Phi(w)$ is standard normal for every $w$. Gaussian
anti-concentration, uniformly over the bounded thresholds
$[0,R/\sigma_0]$, converts the uniform $o_{\Prob}(1)$ remainder into
\eqref{eq:uniformboundarylayer}.
\end{proof}

\begin{figure}[htbp]
\centering
\includegraphics[width=0.98\textwidth]{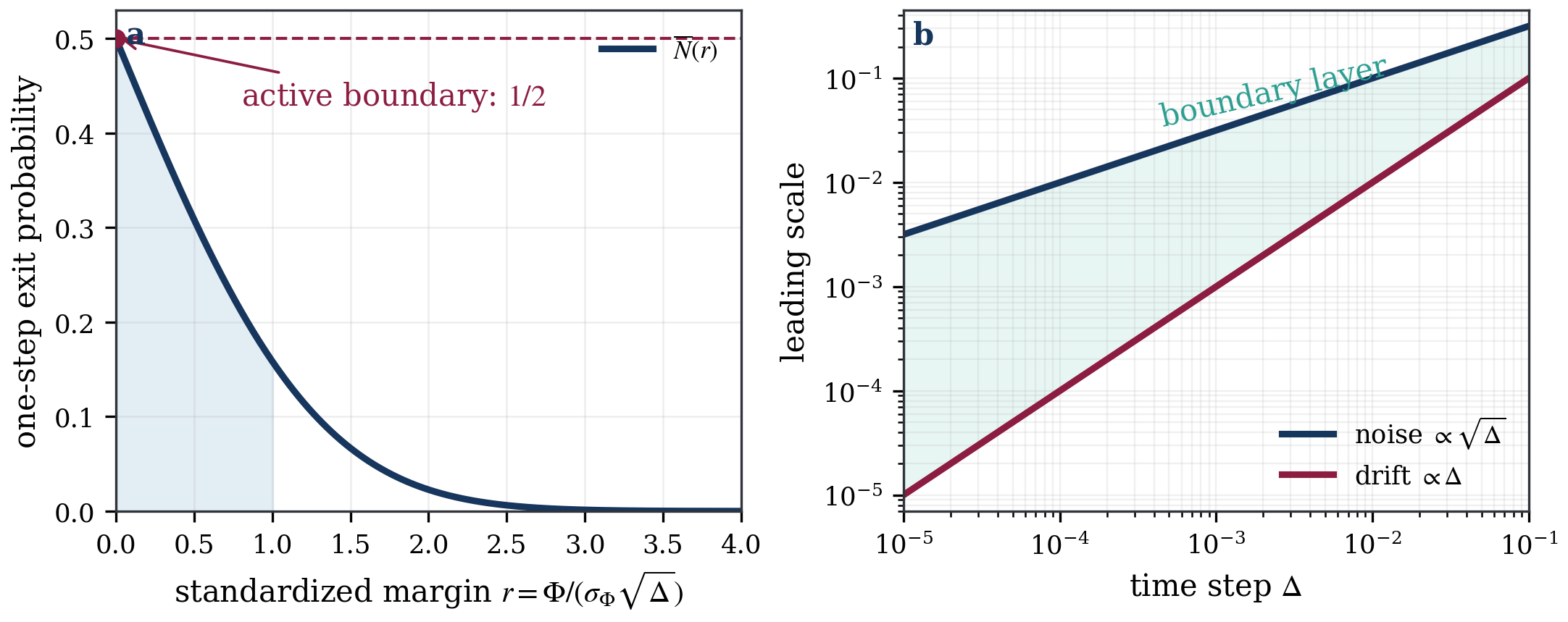}
\caption{The universal Gaussian boundary layer. Panel (a) plots the exact
limit $\overline N(r)$ from Corollary~\ref{cor:boundarylayer}; an active
constraint has standardized margin $r=0$ and exit probability $1/2$. Panel (b)
shows the structural scaling behind the result: normal noise is order
$\sqrt\Delta$, whereas drift is order $\Delta$. These are analytic reference
curves, not estimates.}
\label{fig:boundarylayer}
\end{figure}

The theorem formalizes a simple and unforgiving scaling fact: at the boundary, noise is order $\sqrt{\Delta}$ and drift is order $\Delta$, so no drift --- however cleverly chosen, however dependent on the state --- can rescue a nondegenerate Gaussian model from exiting through an active constraint. The infinitesimal version is sharper still.

\begin{corollary}[Continuous-time trichotomy; \tier{Proved}]\label{cor:trichotomy}
Let $w_t$ solve \eqref{eq:sde} with continuous adapted $\mu$ and $w_t\in\Ksurf_m$ for all $t$ a.s. Let $\Phi$ be one of \eqref{eq:Phi} and let $\theta=\inf\{t:\Phi(w_t)=0\}$ be finite with positive probability. Then on $\{\theta<\infty\}$, a.s.
\[
\sigma_\Phi^2(w_\theta) \;=\; \ip{\ell(w_\theta)}{C\,\ell(w_\theta)}_{\Hs} \;=\; 0 .
\]
Consequently any $\Ksurf_m$-invariant model of type \eqref{eq:sde} must (i) be degenerate in the constraint normal wherever a constraint is active, or leave the class \eqref{eq:sde} by (ii) adding a reflection (Skorokhod) term, or (iii) evolving on a parametrized submanifold contained in $\Ksurf_m$.
\end{corollary}

\begin{proof}
$Y_t=\Phi(w_t)$ is a continuous semimartingale by the Hilbert-space It\^o
formula \citep[Ch.~4]{daprato2014stochastic}; Lemma~\ref{lem:PhiC2} supplies
the required $C^2$ regularity for the butterfly window, and the other two
functionals are linear. Write
\[
 Y_t=Y_0+\int_0^t b_s^\Phi\dd s+M_t,
 \qquad
 \dd\langle M\rangle_t=a_t\dd t,
 \qquad
 a_t=\sigma_\Phi^2(w_t),
\]
where
$b_s^\Phi=D\Phi[w_s]\mu_s+\tfrac12\tr(CD^2\Phi[w_s])$.

Suppose $a_\theta>0$ on a positive-probability subset of
$\{\theta<\infty\}$. Continuity of $w$, $D\Phi$, and the coefficients, followed
by localization and a countable decomposition of that event, supplies
deterministic constants $0<\varepsilon<L<\infty$ and a stopping time
$\zeta>\theta$ such that
\begin{equation}\label{eq:DDSlocalization}
 \varepsilon\le a_t\le L,
 \qquad |b_t^\Phi|\le L,
 \qquad \theta\le t\le\zeta.
\end{equation}
For $0\le u\le\zeta-\theta$, put
$A_u=\langle M\rangle_{\theta+u}-\langle M\rangle_\theta$. Then
$\varepsilon u\le A_u\le Lu$. The Dambis--Dubins--Schwarz theorem
\citep[Thm.~V.1.6--1.7]{revuz1999continuous} gives a Brownian motion
$\widetilde B$ such that
\[
 M_{\theta+u}-M_\theta=\widetilde B_{A_u}.
\]
By the Brownian law of the iterated logarithm at zero
\citep[Ch.~II, Sec.~1]{revuz1999continuous}, there are
$s_n\downarrow0$ for which
$\widetilde B_{s_n}\le-\tfrac12\sqrt{2s_n\log\log(1/s_n)}$.
Because $A$ is continuous and strictly increasing under
\eqref{eq:DDSlocalization}, write $s_n=A_{u_n}$ with $u_n\downarrow0$.
The negative Brownian term is then of order
$-\sqrt{u_n\log\log(1/u_n)}$, while the finite-variation term is bounded below
by $-Lu_n$ and above by $Lu_n$. Hence
$Y_{\theta+u_n}<Y_\theta=0$ for all sufficiently large $n$, contradicting
$Y\ge0$. Therefore $a_\theta=0$ a.s.\ on $\{\theta<\infty\}$. Options (ii)
and (iii) are exactly the two exits from the hypothesis set: (ii) changes the
equation, while (iii) changes the state space.
\end{proof}

Corollary~\ref{cor:trichotomy} is the surface-space analogue of two familiar facts: the CIR square-root diffusion survives at $0$ only because its volatility vanishes there \citep{cox1985theory}, and finite-dimensional HJM realizations live on invariant manifolds \citep{bjork2001existence,filipovic2001consistency}. It also gives the design brief for Section~\ref{sec:constructions}: every viable construction implements one of the three options.

\begin{theorem}[Finite-grid viability test; \tier{Proved, cited}]
\label{thm:gridviability}
Let $x_t\in\R^n$ satisfy $\dd x=b(x)\dd t+\sigma(x)\dd B_t$, with locally
Lipschitz coefficients, and let the discretized arbitrage region be
$K=\{x:\phi_j(x)\ge0,\ j=1,\ldots,J\}$, where the $\phi_j$ are $C^2$ and the
active gradients satisfy the standard constraint qualification. Then $K$ is
stochastically invariant if and only if, for every $x\in K$ and every active
$j$,
\begin{equation}\label{eq:gridviability}
 \sigma(x)^\top\nabla\phi_j(x)=0,
 \qquad
 \nabla\phi_j(x)^\top b(x)
 +\frac12\tr\!\left[\sigma(x)\sigma(x)^\top D^2\phi_j(x)\right]\ge0.
\end{equation}
For affine price-grid constraints the Hessian term vanishes: diffusion must be
tangent and drift inward. For the nonlinear butterfly constraint the second
condition contains the exact It\^o-curvature correction.
\end{theorem}

\begin{proof}
Necessity follows by applying It\^o's formula to each active $\phi_j$: a
nonzero martingale coefficient contradicts one-sided invariance, and after it
vanishes the finite-variation coefficient must point inward. Sufficiency is the
smooth stochastic Nagumo theorem under the stated qualification
\citep{aubin1990stochastic}. Applied to grid evaluations of positivity,
calendar monotonicity, and butterfly inequalities, it gives
\eqref{eq:gridviability}.
\end{proof}

\subsection{Interior survival: a Borell--TIS estimate}\label{sec:interior}

Away from the boundary the news is good, and quantifiably so. The estimate requires the noise to have two derivatives pointwise, which we state as a hypothesis on the model class rather than pretend it follows from $\xi\in\Hs$.

\begin{assumption}[Smooth noise]\label{ass:smooth}
The increment field $\xi\sim N(0,C)$ admits a version with sample paths in
$C^2(\overline{\Dom})$, and
$M\coloneqq\E\norm{\xi}_{C^2(\overline{\Dom})}<\infty$, where
$\norm{h}_{C^2}=\max_{|\beta|\le2}\sup_{\overline{\Dom}}|\partial^\beta h|$.
A concrete sufficient condition is: for some $\alpha>0$, the KL modes satisfy
$e_i\in C^{2,\alpha}(\overline\Dom)$ and
\begin{equation}\label{eq:C2summability}
 \sum_{i\ge1}\sqrt{\lambda_i}\,
 \norm{e_i}_{C^{2,\alpha}(\overline\Dom)}<\infty.
\end{equation}
Indeed, the Gaussian KL series then converges absolutely in mean in
$C^{2,\alpha}$, hence a.s.\ in $C^2$, and $M<\infty$
\citep{adler2007random,bogachev1998gaussian}.
\end{assumption}

\begin{proposition}[Interior survival; \tier{Proved} under Assumption~\ref{ass:smooth}]\label{prop:survival}
Let $w_0\in\Ksurf_m\cap C^2(\overline{\Dom})$ have joint margin
\[
\delta \;\coloneqq\; \min\Bigl\{\,1,\ \inf_{\overline{\Dom}}(w_0-m),\ \operatorname*{ess\,inf}_{\Dom}\partial_\tau w_0,\ \ L^{-1}\operatorname*{ess\,inf}_{\Dom} g[w_0]\Bigr\} \;>\;0,
\]
where $L$ is a Lipschitz constant of $h\mapsto g[w_0+h]$ from the unit $C^2$-ball into $L^\infty$ (finite because $g$ is a polynomial in $(1/w,\,k,\,w_k,\,w_{kk})$ and $w_0+h\ge m$ there for $\delta\le\inf(w_0-m)$). Let $\bar\sigma^2\coloneqq\sup\{\E\,\lambda(\xi)^2:\ \lambda\in(C^2)^\ast,\ \norm{\lambda}\le1\}$. Then under Assumption~\ref{ass:smooth}, provided $\delta>M$,
\begin{equation}\label{eq:borell}
\Prob\bigl[\,w_0+\xi\notin\Ksurf_m\,\bigr] \;\le\; \Prob\bigl[\norm{\xi}_{C^2}>\delta\bigr] \;\le\; \exp\!\Bigl(-\frac{(\delta-M)^2}{2\bar\sigma^2}\Bigr).
\end{equation}
\end{proposition}

\begin{proof}
If $\norm{\xi}_{C^2}\le\delta\le1$ then: $w_0+\xi\ge m+\inf(w_0-m)-\delta\ge m$; $\partial_\tau(w_0+\xi)\ge\operatorname{ess\,inf}\partial_\tau w_0-\delta\ge0$ a.e.; and $g[w_0+\xi]\ge g[w_0]-L\norm{\xi}_{C^2}\ge0$ a.e. This proves the first inequality. The second is the Borell--TIS inequality for the norm of a centered Gaussian element of the separable Banach space $C^2(\overline{\Dom})$ \citep{borell1975brunn,adler2007random}: $\Prob[\norm{\xi}\ge \E\norm{\xi}+t]\le\exp(-t^2/(2\bar\sigma^2))$ with $t=\delta-M$.
\end{proof}

Together, Theorem~\ref{thm:boundary} and Proposition~\ref{prop:survival} say that the Gaussian factor picture is a boundary-layer approximation: exponentially safe at margin $\delta\gg M$, exactly half-wrong at margin $0$. Empirically, index surfaces spend calm regimes in the interior and stress regimes with the butterfly constraint nearly active in the short-dated wings --- which is where Section~\ref{sec:empirics} predicts unconstrained generators fail (hypothesis H1).

\section{Infinite-dimensional dynamics and exact modal reduction}\label{sec:hilbertdynamics}

The surface has one genuine time variable: calendar time $t$. Log-moneyness and
time to maturity are state coordinates, not additional filtration directions.
Accordingly, the natural dynamic object is an $\Hs$-valued It\^o process
\begin{equation}\label{eq:hilbertSDE}
 \dd w_t=(\mathcal A w_t+F(w_t))\dd t+B(w_t)\dd W_t,
\end{equation}
where $\mathcal A$ generates a strongly continuous semigroup, $W$ is a cylindrical or
$Q$-Wiener process on a separable noise space, and $B(w)$ is Hilbert--Schmidt
from the Cameron--Martin space into $\Hs$.  Its mild form is
\begin{equation}\label{eq:mild}
w_t=S(t)w_0+\int_0^tS(t-s)F(w_s)\dd s
       +\int_0^tS(t-s)B(w_s)\dd W_s .
\end{equation}
In the time-to-maturity parametrization, the generator $\mathcal A$ must include
the Musiela transport $\partial_\tau$ from \eqref{eq:musiela}; the remaining
drift is not arbitrary when \eqref{eq:hilbertSDE} is intended as a pricing
model, because Proposition~\ref{prop:musiela} must then hold jointly with the
forward and discount-factor dynamics. For physical-measure forecasting,
$F$ may instead be estimated statistically, while static feasibility is
enforced separately.
This formulation distinguishes two questions that are often conflated.  The
first is whether a function-valued surface exists; the second is whether a few
factors approximate it accurately.  Trace-class smoothing answers the first,
while the following theorem answers the second in the benchmark case.

\begin{assumption}[Diagonal linear surface dynamics]\label{ass:diagonal}
There is an orthonormal basis $(e_i)_{i\ge1}$ of $\Hs$ such that
$\mathcal A e_i=-\alpha_i e_i$, $\alpha_i\ge0$, and $Qe_i=q_i e_i$, $q_i\ge0$.
The dynamics are $\dd w_t=\mathcal A w_t\dd t+\dd W_t^Q$, and
\begin{equation}\label{eq:summability}
 \sum_{i\ge1}\frac{q_i}{1+\alpha_i}<\infty .
\end{equation}
\end{assumption}

\begin{theorem}[Exact modal SDE and truncation error; \tier{Proved}]
\label{thm:modalSDE}
Under Assumption~\ref{ass:diagonal}, write $x_i(t)=\ip{w_t}{e_i}_{\Hs}$ and
$x_i(0)=\ip{w_0}{e_i}_{\Hs}$. Then
\begin{equation}\label{eq:modeSDE}
 \dd x_i(t)=-\alpha_i x_i(t)\dd t+\sqrt{q_i}\dd\beta_i(t),
 \qquad
 x_i(t)=e^{-\alpha_it}x_i(0)+\sqrt{q_i}\int_0^t e^{-\alpha_i(t-s)}\dd\beta_i(s),
\end{equation}
for independent Brownian motions $(\beta_i)$.  For every $m$, the vector
$(x_1,\ldots,x_m)$ is an exact $m$-dimensional SDE and
\begin{equation}\label{eq:dynamictruncation}
 \E\left\|w_t-\sum_{i=1}^m x_i(t)e_i\right\|_{\Hs}^{2}
 =\sum_{i>m}e^{-2\alpha_it}|x_i(0)|^2
  +\sum_{i>m}q_i\rho_i(t),
 \qquad
 \rho_i(t)=\begin{cases}
 (1-e^{-2\alpha_it})/(2\alpha_i),&\alpha_i>0,\\ t,&\alpha_i=0.
 \end{cases}
\end{equation}
The stochastic convolution is $\Hs$-valued for every $t>0$ if and only if
$\sum_iq_i\rho_i(t)<\infty$, which is equivalent to
\eqref{eq:summability}.
\end{theorem}

\begin{proof}
Expand $W_t^Q=\sum_i\sqrt{q_i}\,\beta_i(t)e_i$ and project the mild equation
onto $e_i$.  This gives \eqref{eq:modeSDE}. Orthogonality and It\^o isometry
give \eqref{eq:dynamictruncation}.  For fixed $t>0$ there are constants
$0<c_t\le C_t<\infty$ such that
$c_t/(1+\alpha_i)\le\rho_i(t)\le C_t/(1+\alpha_i)$: use
$1-e^{-u}\le\min\{u,1\}$ for the upper bound and split the lower bound at
$\alpha_i=1$ (continuity supplies the constant on $[0,1]$, while
$1-e^{-2\alpha_it}\ge1-e^{-2t}$ for $\alpha_i\ge1$). The stated equivalence follows.
\end{proof}

\begin{remark}[What a factor model asserts]\label{rem:factorassertion}
Theorem~\ref{thm:modalSDE} makes the approximation content explicit. A
finite-factor surface model is not merely a regression: it chooses a subspace,
a stochastic law within that subspace, and an omitted-energy budget. If the
factor manifold is also required to lie in $\Ksurf$, the boundary conditions of
Section~\ref{sec:dynamics} constrain its diffusion coefficients. Low rank alone
does not imply no-arbitrage.
\end{remark}

\begin{remark}[Almost-diagonal stability; \tier{Proved, cited}]
\label{rem:almostdiagonal}
Assumption~\ref{ass:diagonal} is a transparent benchmark, not a qualitative
boundary. It is a statistical benchmark on the fixed $(k,\tau)$ grid, not a
diagonalization of the transported pricing generator: the Musiela shift
$\partial_\tau$ is not diagonalizable in the orthonormal dissipative form
postulated by the assumption. Let $Q_0$ be the diagonal part of $Q$ in the eigenbasis of
$\mathcal A$, and suppose the semigroup is contractive. Coupling the two
stochastic convolutions with the same cylindrical noise and using the
Powers--St\o rmer inequality gives
\[
\E\left\|\int_0^tS(t-s)(Q^{1/2}-Q_0^{1/2})\dd W_s\right\|_{\Hs}^2
\le t\|Q^{1/2}-Q_0^{1/2}\|_{\mathrm{HS}}^2
\le t\|Q-Q_0\|_1.
\]
Thus the diagonal reduction is stable in trace norm. On a finite spectral
truncation, gaps of $\mathcal A$ convert the commutator identity
$[\mathcal A,Q]_{ij}=(\alpha_j-\alpha_i)Q_{ij}$ into explicit off-diagonal
bounds \citep{powers1970free}.
\end{remark}

For nonlinear $F$ and state-dependent $B$, the coordinates are coupled and a
Galerkin system is an approximation rather than an exact subsystem. Standard
Hilbert-space SDE theory supplies convergence under Lipschitz or monotonicity
conditions \citep{daprato2014stochastic}. The geometry developed above adds the
separate viability requirement: on every active smooth constraint $\Phi$, the
diffusion must satisfy $B(w)^\ast D\Phi(w)=0$, with the corresponding inward
drift condition after the It\^o correction. This is the analytic bridge from
infinite-dimensional existence to arbitrage-free existence.

\section{Spectral structure: covariance operators and Karhunen--Lo\`eve factors}\label{sec:spectral}

\subsection{The increment covariance operator}\label{sec:covop}

Model daily increments $\Delta w_t = w_{t+1}-w_t$ (or log-increments; the choice is part of the empirical design) as centered second-order random elements of $\Hs$ with covariance operator
\begin{equation}\label{eq:cov}
C:\Hs\to\Hs,\qquad \ip{u}{Cv}_{\Hs} \;=\; \E\,\ip{u}{\Delta w}_{\Hs}\,\ip{v}{\Delta w}_{\Hs}.
\end{equation}
$C$ is self-adjoint, nonnegative, and trace class with $\tr C=\E\norm{\Delta w}^2_{\Hs}$ \citep{daprato2014stochastic,bogachev1998gaussian}. The spectral theorem yields eigenpairs $(\lambda_i,e_i)_{i\ge1}$, $\lambda_1\ge\lambda_2\ge\cdots\ge0$, $\sum_i\lambda_i<\infty$, orthonormal in $\Hs$.

\begin{theorem}[Karhunen--Lo\`eve expansion; \tier{Proved}]\label{thm:kl}
With $(\lambda_i,e_i)$ as above, define, for $\lambda_i>0$,
\[
Z_i=\lambda_i^{-1/2}\ip{e_i}{\Delta w}_{\Hs}.
\]
\[
\Delta w \;=\; \sum_{i\ge1}\sqrt{\lambda_i}\,Z_i\,e_i,
\]
with $(Z_i)$ uncorrelated, mean zero, unit variance, convergence in $L^2(\Omega;\Hs)$, and truncation error $\E\norm{\Delta w-\sum_{i\le q}\sqrt{\lambda_i}Z_ie_i}^2_{\Hs}=\sum_{i>q}\lambda_i$. If $\Delta w$ is Gaussian the $Z_i$ are i.i.d.\ standard normal and convergence also holds a.s.
\end{theorem}

\begin{figure}[htbp]
\centering
\includegraphics[width=0.98\textwidth]{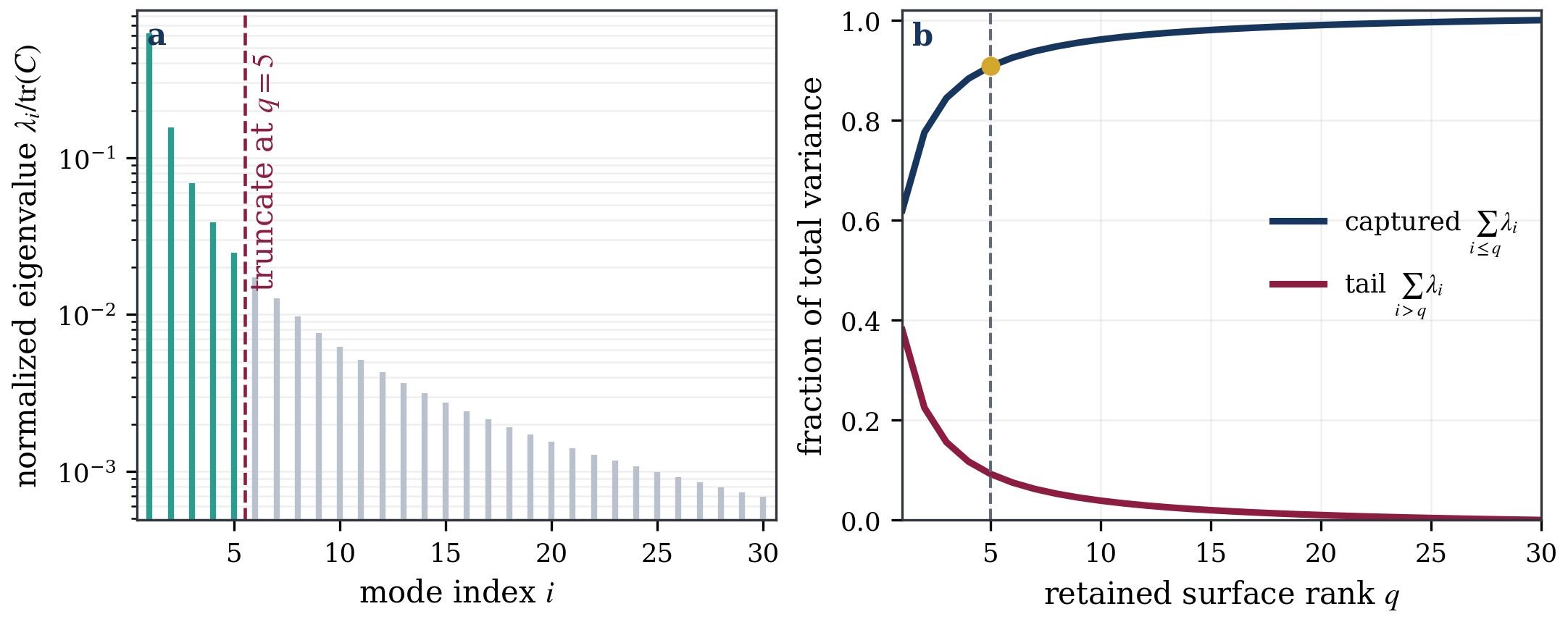}
\caption{Spectral truncation and its auditable error. The display uses the
normalized reference spectrum $\lambda_i\propto i^{-2}$ only to visualize
Theorem~\ref{thm:kl}: retaining rank $q$ captures
$\sum_{i\le q}\lambda_i$ and leaves the exact mean-square tail
$\sum_{i>q}\lambda_i$. It is not an empirical eigenspectrum.}
\label{fig:spectralbudget}
\end{figure}

Three remarks tie the abstract statement to surface practice. First, \emph{the modes are the factors}: level, term-structure tilt, skew, and smile curvature emerge as the leading $e_i$ in every published surface PCA since \citet{skiadopoulos1999dynamics} and \citet{cont2002dynamics}; Theorem~\ref{thm:kl} equips them with a norm, an orthogonality, and an error bound that grid PCA lacks. Second, \emph{the weight matters}: on a discrete quote grid $\{x_j\}$ with quadrature weights $\{q_j\}$ and Sobolev weight $\rho$, consistent estimation of $(\lambda_i,e_i)$ is PCA of the matrix $S^{1/2}\widehat\Sigma S^{1/2}$ with $S=\operatorname{diag}(q_j\rho(x_j))$ extended to derivative penalties --- i.e., functional PCA \citep{ramsay2005functional,horvath2012inference}, not raw PCA; ignoring $S$ tilts eigenvectors toward densely quoted short-dated strikes. Third, \emph{estimation is its own problem}: $\widehat C$ inherits every pathology of large covariance estimation, and the repair taxonomy of \citet{noguer2026covariance} --- which remedy fixes which failure --- applies verbatim to the operator setting, with the market-mode dominance analysis of \citet{noguer2026market} describing the $\lambda_1$-gap regime that surface increments empirically occupy.

\subsection{Constraints and the spectral picture}\label{sec:constraintspectral}

The KL expansion is an interior object: it diagonalizes the unconditional second moment, while the conditional law near $\partial\Ksurf_m$ is non-Gaussian by Theorem~\ref{thm:boundary} no matter what the modes are. The reconciliation is quantitative, not qualitative:

\begin{corollary}[Constraint-consistent truncation; \tier{Proved} under Assumption~\ref{ass:smooth}]\label{cor:truncation}
Let $\xi^{(q)}=\sum_{i\le q}\sqrt{\lambda_i}Z_ie_i$ be a Gaussian KL truncation whose modes satisfy the sufficient condition of Assumption~\ref{ass:smooth}, and let $w_0$ have margin $\delta$ as in Proposition~\ref{prop:survival} with $M_q\coloneqq\E\norm{\xi^{(q)}}_{C^2}<\delta$. Then $\Prob[w_0+\xi^{(q)}\notin\Ksurf_m]\le\exp(-(\delta-M_q)^2/(2\bar\sigma_q^2))$, with $M_q$ and $\bar\sigma_q$ nondecreasing in $q$.
\end{corollary}

Truncation is thus a \emph{safety parameter}: fewer modes mean smaller $M_q$, hence exponentially fewer violations --- at the price of the reconstruction error $\sum_{i>q}\lambda_i$. Hypothesis H3 of Section~\ref{sec:empirics} makes the resulting trade-off falsifiable.

\section{Functional Greeks and minimum-variance field hedging}\label{sec:hedging}

\subsection{The vega field}\label{sec:vegafield}

Let $\Pi:\mathcal{U}\subset\Hs\to\R$ be the mark-to-market value of a book as a functional of the surface, defined and Fr\'echet differentiable on an open set $\mathcal{U}\supset\Ksurf_m$ of surfaces (for books of European payoffs priced by the quoted surface, differentiability follows from \eqref{eq:vega} and the chain rule; for path-dependent books it is a model property).

\begin{definition}[Vega field]\label{def:vegafield}
The \emph{vega field} of $\Pi$ at $w$ is the Riesz representer $\nu_w\in\Hs$ of $D\Pi[w]$:
\begin{equation}\label{eq:vegafielddef}
D\Pi[w]h \;=\; \ip{\nu_w}{h}_{\Hs} \qquad \text{for all } h\in\Hs .
\end{equation}
\end{definition}

The vega field is one object where the desk keeps a spreadsheet of buckets, and the two are reconciled by the kernel structure of $\Hs$:

\begin{lemma}[Bucketed vegas are samples of the field; \tier{Proved}]\label{lem:bucket}
Let $R_x$ be the kernel sections of Proposition~\ref{prop:embedding}(ii) and define the \emph{kernel bump} at node $x_j$ as the surface perturbation $h_j=R_{x_j}$. Then the bucketed vega with respect to kernel bumps equals the field sampled at the node:
\[
D\Pi[w]\,R_{x_j} \;=\; \ip{\nu_w}{R_{x_j}}_{\Hs} \;=\; \nu_w(x_j).
\]
\end{lemma}

\begin{proof}
The first equality is \eqref{eq:vegafielddef}; the second is the reproducing property.
\end{proof}

\begin{lemma}[Covariance under a change of Sobolev weight; \tier{Proved}]
\label{lem:weightcovariance}
Let $\ip{\cdot}{\cdot}_1$ and $\ip{\cdot}{\cdot}_2$ be the equivalent
$H^2$ inner products generated by two admissible weights $\rho_1,\rho_2$.
There is a unique bounded, positive, self-adjoint, invertible operator $T$ on
$(\Hs,\ip{\cdot}{\cdot}_1)$ such that
\begin{equation}\label{eq:weightmetric}
 \ip{u}{v}_2=\ip{Tu}{v}_1.
\end{equation}
For the same derivative covector $D\Pi[w]$, kernel evaluation, and random
increment $X$, let $\nu_i,R_x^{(i)},C_i$ denote their Riesz representers and
covariance operators in metric $i$. Then
\begin{equation}\label{eq:weighttransforms}
 \nu_2=T^{-1}\nu_1,\qquad
 R_x^{(2)}=T^{-1}R_x^{(1)},\qquad
 C_2=C_1T.
\end{equation}
If every book and instrument sensitivity is transformed accordingly, both the
hedge weights \eqref{eq:alphastar} and the residual variance
\eqref{eq:residual} are unchanged.
\end{lemma}

\begin{proof}
The Riesz theorem applied to the identity map between the two equivalent
Hilbert norms gives $T$ and \eqref{eq:weightmetric}. Since
$D\Pi[w]h=\ip{\nu_1}{h}_1=\ip{T\nu_2}{h}_1$, one has
$\nu_1=T\nu_2$; the same calculation for evaluation gives
$R_x^{(1)}=TR_x^{(2)}$. Moreover,
\[
 \ip{u}{C_2v}_2
 =\E\ip{u}{X}_2\ip{v}{X}_2
 =\ip{Tu}{C_1Tv}_1
 =\ip{u}{C_1Tv}_2,
\]
so $C_2=C_1T$. Let $H_i^{\ast_i}$ denote the adjoint in metric $i$.
Since $H_2=T^{-1}H_1$ by the same Riesz argument,
$H_2^{\ast_2}C_2H_2=H_1^{\ast_1}C_1H_1$ and
$H_2^{\ast_2}C_2\nu_2=H_1^{\ast_1}C_1\nu_1$; substitution into
\eqref{eq:alphastar}--\eqref{eq:residual} proves invariance.
\end{proof}

Lemma~\ref{lem:bucket} says the familiar bucket report is a pointwise
discretization of a well-defined function, \emph{provided} the bump shapes are
the kernel sections of the chosen inner product. For arbitrary linearly
independent bumps $b_1,\ldots,b_J$, the reported vector
$\beta_j=D\Pi[w]b_j$ and the coefficients $a$ of the projected representer
$\nu_B=\sum_i a_i b_i$ satisfy
\begin{equation}\label{eq:bumpgram}
 \beta=G_Ba,\qquad (G_B)_{ij}=\ip{b_i}{b_j}_{\Hs}.
\end{equation}
Thus triangular buckets are Gram coordinates; kernel buckets are literal
samples. Lemma~\ref{lem:weightcovariance} shows how the field and covariance
change with $\rho$. The weight changes their representation, not the underlying
derivative or optimal hedge.

\subsection{The minimum-variance field hedge}\label{sec:mvhedge}

Hedge with $n$ traded instruments (options, VIX futures, variance swaps) whose values $P_i$ are Fr\'echet differentiable in $w$ with sensitivity fields $\eta_i=$ Riesz representers of $DP_i[w]$. Assemble the linear map
\[
H:\R^n\to\Hs,\qquad H\alpha=\sum_{i=1}^n\alpha_i\,\eta_i,
\qquad\text{with adjoint}\qquad
H^\ast:\Hs\to\R^n,\quad (H^\ast u)_i=\ip{\eta_i}{u}_{\Hs}.
\]
To first order, the one-period P\&L of the hedged book against a surface increment $\Delta w$ with covariance operator $C$ is $\ip{\nu-H\alpha}{\Delta w}_{\Hs}$, with variance
\begin{equation}\label{eq:variance}
V(\alpha) \;=\; \ip{\,\nu-H\alpha\,}{\;C\,(\nu-H\alpha)\,}_{\Hs}.
\end{equation}

\begin{theorem}[Minimum-variance field hedge; \tier{Proved}]\label{thm:hedge}
Let $G\coloneqq H^\ast C H\in\R^{n\times n}$ (the Gram matrix of the instrument fields in the $C$-semi-inner-product $\ip{u}{v}_C\coloneqq\ip{u}{Cv}_{\Hs}$). If $G\succ0$ --- equivalently, the fields $C^{1/2}\eta_1,\dots,C^{1/2}\eta_n$ are linearly independent --- then $V$ has the unique minimizer
\begin{equation}\label{eq:alphastar}
\boxed{\ \alpha^\ast \;=\; G^{-1}\,H^\ast C\,\nu\ }
\end{equation}
with minimized residual variance
\begin{equation}\label{eq:residual}
V(\alpha^\ast) \;=\; \ip{\nu}{C\nu}_{\Hs} \;-\; \bigl(H^\ast C\nu\bigr)^{\!\top} G^{-1}\bigl(H^\ast C\nu\bigr) \;\ge\;0 .
\end{equation}
Moreover $H\alpha^\ast$ is the $\ip{\cdot}{\cdot}_C$-orthogonal projection of $\nu$ onto $\operatorname{span}\{\eta_i\}$, and \eqref{eq:residual} vanishes iff $C^{1/2}\nu\in\operatorname{span}\{C^{1/2}\eta_i\}$.
\end{theorem}

\begin{proof}
$V(\alpha)=\norm{C^{1/2}(\nu-H\alpha)}^2_{\Hs}$ is a convex quadratic in $\alpha$ with Hessian $2G\succ0$; the first-order condition $H^\ast C(\nu-H\alpha)=0$ gives $G\alpha=H^\ast C\nu$, hence \eqref{eq:alphastar}, and substituting back gives \eqref{eq:residual}. The projection statement is the normal-equation characterization of least squares in the semi-inner product $\ip{\cdot}{\cdot}_C$.
\end{proof}

\begin{figure}[htbp]
\centering
\includegraphics[width=0.98\textwidth]{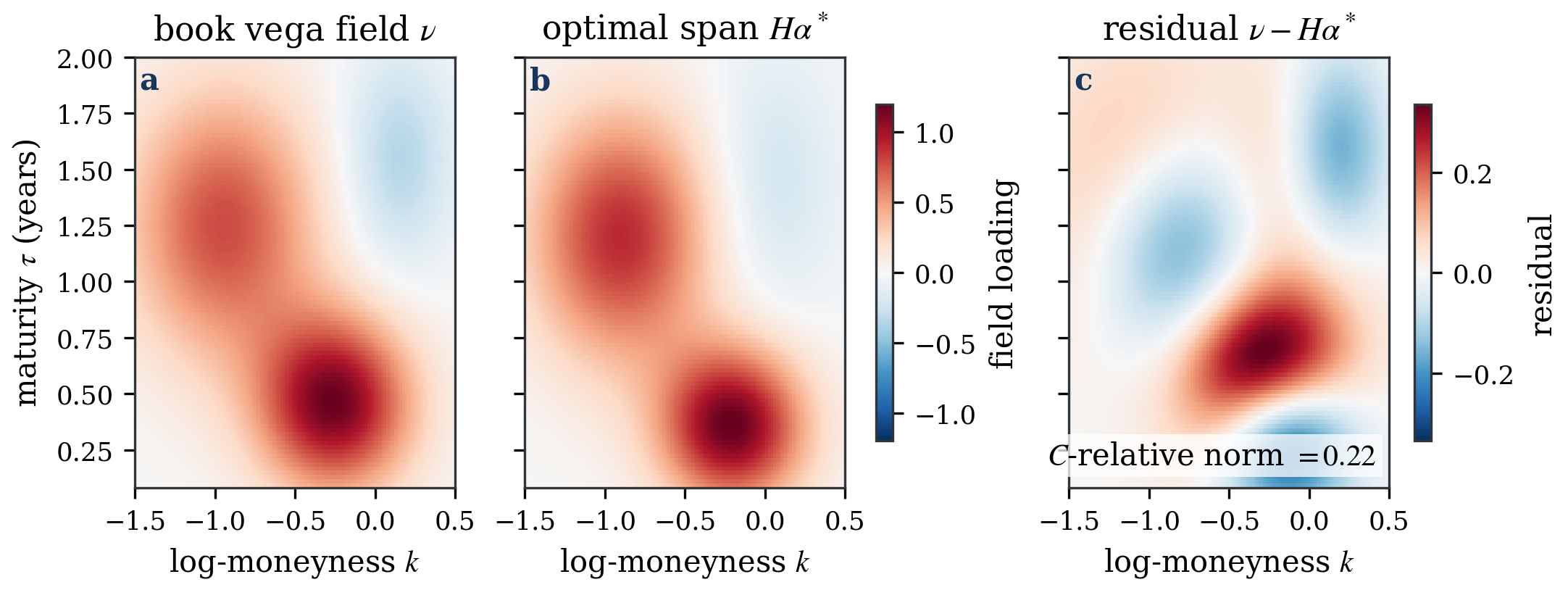}
\caption{Field hedging as covariance-weighted projection. A deterministic
schematic book field $\nu$ (left) is projected onto three smooth instrument
fields to obtain $H\alpha^\ast$ (center); the unspanned component is the
residual field (right). The displayed relative norm uses a stylized diagonal
$C$-metric emphasizing short maturities. The construction illustrates
Theorem~\ref{thm:hedge} and contains no fitted market observations.}
\label{fig:fieldhedge}
\end{figure}

\begin{remark}[Reading the formula]\label{rem:hedgeread}
Formula \eqref{eq:alphastar} is generalized least squares of the vega field on the instrument fields, with the market's own increment covariance as the metric. Its content relative to bucket practice: (i) hedge weights should load on instruments whose sensitivity fields align with $\nu$ \emph{in the directions the surface actually moves} --- a short-dated wing option is a poor hedge for a long-dated ATM exposure even if their buckets overlap, because $C$ assigns those directions nearly orthogonal dynamics; (ii) the residual \eqref{eq:residual} is a computable lower bound on achievable vega P\&L variance, hence an audit statistic (hypothesis H4); (iii) when instruments are $C$-collinear, $G$ is ill-conditioned and \eqref{eq:alphastar} should be regularized, $\alpha^\ast_\lambda=(G+\lambda I)^{-1}H^\ast C\nu$ --- and which regularization to choose is precisely the operator-repair question catalogued in \citet{noguer2026covariance}. The joint hedge including the underlying (spot delta alongside the field) is a two-block version of the same normal equations, Appendix~\ref{app:joint}.
\end{remark}

\section{Three viable constructions}\label{sec:constructions}

Corollary~\ref{cor:trichotomy} leaves exactly three doors open. Each of the following templates walks through one of them.

\subsection{Construction I: price-coordinate dynamics (convexity door)}\label{sec:constructionI}

Fix an economically nonbinding upper variance cap $\Mcap>m$ and work in the banded
price set $\Kcmbar$. Proposition~\ref{prop:chartcorrespondence} makes this set
closed and convex while keeping $B^{-1}$ uniformly regular. Pose reflected
(Skorokhod-type) dynamics
\begin{equation}\label{eq:reflected}
\dd c_t \;=\; \mu^c_t\,\dd t \;+\; \Sigma^c\,\dd W_t \;+\; \dd\eta_t,
\qquad c_t\in\Kcmbar,\qquad
\eta \ \text{of bounded variation, } \dd\eta_t\in -N_{\Kcmbar}(c_t),
\end{equation}
with $N_{\Kcmbar}$ the (convex-analysis) normal cone; the reflection term acts
only when a constraint is active and pushes inward. The surface is recovered
uniquely as $w_t=B^{-1}[c_t]\in\Ksurf_{m,\Mcap}$, and a martingale drift restriction
on $\mu^c$ (each $c_t(k,\tau)$ deflated appropriately) connects the construction
to the market-model literature
\citep{schonbucher1999market,schweizer2008term,carmona2009local,kallsen2015hjm}.

\begin{proposition}[Status; \tier{Conditional}]\label{prop:reflectedstatus}
Existence and uniqueness theories for stochastic variational inequalities and reflected equations on closed convex sets in Hilbert space --- \citet{haussmann1989stochastic} for the parabolic variational setting, \citet{nualart1992reflection} for reflected SPDEs, \citet{cepa1998skorohod} and \citet{slominski2001euler} in finite dimensions --- cover equations of the form \eqref{eq:reflected} under structural hypotheses on the set and coefficients. We tier well-posedness for the specific set $\Kcmbar$ \tier{Conditional} on verifying those hypotheses (notably the characterization of $N_{\Kcmbar}$ generated by the static and band constraints, and the compatibility of $\Sigma^c$ with them); the verification is a self-contained technical project that this paper scopes but does not execute. The band removes the separate inverse-map pathology: no reflected path can be pinned at the intrinsic boundary or at $c=1$.
\end{proposition}

\begin{theorem}[Finite-grid reflected price realization; \tier{Proved, cited}]
\label{thm:gridreflection}
Fix a rectangular quote grid $\Gamma=\{(k_i,\tau_i)\}_{i=1}^N$. Let
$D_\tau^\Gamma$ and $L_k^\Gamma$ be fixed consistent difference matrices for
$\partial_\tau$ and $\partial_{kk}-\partial_k$, and define
\begin{equation}\label{eq:gridpriceset}
K_{c,m,\Mcap}^\Gamma=\left\{x\in\R^N:
B(k_i,m)\le x_i\le B(k_i,\Mcap),\quad
D_\tau^\Gamma x\ge0,\quad L_k^\Gamma x\ge0\right\}.
\end{equation}
Assume $K_{c,m,\Mcap}^\Gamma$ has nonempty interior, let $b$ and $\sigma$ be
globally Lipschitz, and let $X_0\in K_{c,m,\Mcap}^\Gamma$. Then this set is a
compact convex polytope, and the normally reflected equation
\begin{equation}\label{eq:gridreflected}
\dd X_t=b(X_t)\dd t+\sigma(X_t)\dd B_t+\dd\eta_t,\qquad
\dd\eta_t\in-N_{K_{c,m,\Mcap}^\Gamma}(X_t),
\end{equation}
has a unique strong solution. Moreover, the projection Euler scheme
\begin{equation}\label{eq:gridprojectedeuler}
X_{r+1}^{\Delta}=\Pi_{K_{c,m,\Mcap}^\Gamma}
\left(X_r^{\Delta}+b(X_r^{\Delta})\Delta
+\sigma(X_r^{\Delta})\Delta B_r\right)
\end{equation}
converges to that solution uniformly on compact time intervals in probability.
At every node, $w_i=B^{-1}(k_i,X_i)$ is unique and lies in
$[m,\Mcap]$; hence the resulting grid satisfies the declared discrete static
constraints exactly, at every time and every discretization step.
\end{theorem}

\begin{proof}
The set in \eqref{eq:gridpriceset} is a finite intersection of closed
half-spaces and is bounded by its Black--Scholes price band, hence it is a
compact convex polytope. Its normal cone is maximal monotone. The multivalued
Skorokhod theorem of \citet{cepa1998skorohod} gives strong existence and
pathwise uniqueness for \eqref{eq:gridreflected}; the projected Euler
convergence is the finite-dimensional result of \citet{slominski2001euler}.
Strict positivity of $B_w$ on the compact band gives the final nodewise inverse
by Proposition~\ref{prop:chartcorrespondence}. This is a finite-grid static
realization; the fixed-contract martingale restriction of
Proposition~\ref{prop:musiela} remains a separate requirement.
\end{proof}

The construction's virtue is exactness: no arbitrage is enforced pathwise by geometry, not approximately by penalty. Theorem~\ref{thm:gridreflection} closes the existence and numerical-convergence argument on every fixed quote grid; Proposition~\ref{prop:reflectedstatus} isolates the remaining function-space verification. Its costs are the nonlinearity of $B^{-1}$ (Gaussian price factors are non-Gaussian vol factors) and the opacity of the reflection term to economic interpretation.

\subsection{Construction II: tangent-projected innovations}\label{sec:constructionII}

Keep total-variance coordinates and Gaussian proposals, but project each proposed innovation onto the feasible cone before applying it. Given $w_j\in\Ksurf_m$, a step $\Delta$, and a KL proposal $\widehat\delta_j=\Delta\mu(w_j)+\sqrt{\Delta}\sum_{i\le q}\sqrt{\lambda_i}Z_{ij}e_i$, set
\begin{equation}\label{eq:qp}
\delta_j \;=\; \operatorname*{arg\,min}_{\delta\in\Hs_q}\ \norm{\delta-\widehat\delta_j}_{\Hs}^2
\quad\text{s.t.}\quad
\begin{cases}
w_j+\delta \ge m \ \text{on } \overline{\Dom},\\
\partial_\tau(w_j+\delta)\ge 0 \ \text{a.e.},\\
g[w_j] + Dg[w_j]\,\delta \ge -\epsilon_j \ \text{a.e.},
\end{cases}
\qquad w_{j+1}=w_j+\delta_j,
\end{equation}
where $\Hs_q=\operatorname{span}\{e_1,\dots,e_q\}$ (plus a spline correction space if desired), the butterfly constraint is linearized via Lemma~\ref{lem:Dg}, and $\epsilon_j\downarrow0$ is a slack schedule absorbing the linearization error of Appendix~\ref{app:Dg}. On a quote grid, \eqref{eq:qp} is a finite-dimensional convex quadratic program: the objective is quadratic, and all constraints are affine in $\delta$. When the linearized step overshoots (the true $g[w_{j+1}]$ dips below $0$ within tolerance), an inner correction re-solves \eqref{eq:qp} at $w_{j+1}$ with $\widehat\delta=0$ --- a projection step onto the linearized feasible set.

\begin{proposition}[Status]\label{prop:qpstatus}
(i) \tier{Proved}: each QP \eqref{eq:qp} on a finite grid has a strictly convex objective and affine constraints, hence a unique solution whenever feasible (and $\delta=$ projection of $0$ restores feasibility), and the scheme preserves the linearized constraints by construction. (ii) \tier{Conditional}: in price coordinates, projected-Euler schemes for reflected diffusions on convex sets converge to the reflected dynamics of Construction I as $\Delta\downarrow0$; the finite-dimensional theory is \citet{slominski2001euler}, and the tier is conditional on its infinite-dimensional extension under the hypotheses of Proposition~\ref{prop:reflectedstatus}. (iii) \tier{Conjectural}: in total-variance coordinates, we conjecture the scheme converges weakly to a degenerate-diffusion limit whose noise operator annihilates the active constraint normals --- the object described abstractly by Corollary~\ref{cor:trichotomy}(i); we know of no off-the-shelf theorem covering the nonconvex, second-order-constraint case, and flag it as the paper's main open problem.
\end{proposition}

Construction II is the practitioner's construction: it retains the KL factors of Section~\ref{sec:spectral} (hence interpretability and calibration to historical increments), touches proposals only near the boundary (Proposition~\ref{prop:survival} quantifies how rarely), and costs one small QP per step.

\subsection{Construction III: arbitrage-free parameter maps (submanifold door)}\label{sec:constructionIII}

Choose a smooth parametrization $\varphi:\Theta\subset\R^p\to\Hs$ whose image lies in $\Ksurf$ by construction, and diffuse the parameters:
\begin{equation}\label{eq:paramsde}
\dd\theta_t=b(\theta_t)\dd t+\gamma(\theta_t)\dd B_t \ \ \text{in }\Theta_{\mathrm{adm}},\qquad w_t=\varphi(\theta_t).
\end{equation}
The canonical example is the SSVI family of \citet{gatheral2014arbitrage}: with ATM total variance term structure $\theta_\tau$ and smile function $\phi$,
\[
w(k,\tau)=\frac{\theta_\tau}{2}\Bigl(1+\rho\,\phi(\theta_\tau)k+\sqrt{(\phi(\theta_\tau)k+\rho)^2+1-\rho^2}\Bigr),
\]
with sufficient no-butterfly conditions of the form $\theta_\tau\phi(\theta_\tau)(1+|\rho|)<4$ and $\theta_\tau\phi(\theta_\tau)^2(1+|\rho|)\le4$, and calendar conditions on the monotonicity of $\theta_\tau$ and the behavior of $\phi$ \citep[Thm.~4.2, Lem.~4.3]{gatheral2014arbitrage}; the eSSVI extension of \citet{hendriks2019essvi} allows maturity-dependent $\rho$. Keeping $\theta_t$ inside $\Theta_{\mathrm{adm}}$ is a finite-dimensional invariance problem (log/logit coordinates, or one-dimensional reflection), which is standard. Neural parametrizations with certified no-arbitrage layers \citep{cohen2023neural,cont2023simulation,cuchiero2020gan} are Construction III with $p$ large and $\varphi$ learned.

The construction's virtue is tractability; its cost is rigidity. The image $\varphi(\Theta)$ is a $p$-dimensional submanifold of an infinite-dimensional set: the model asserts that surface innovations have exactly rank $p$, with shapes confined to $\operatorname{ran}D\varphi$. Hypothesis H5 tests precisely this assertion's failure mode --- autocorrelated residual fields off the manifold.

\subsection{Comparison}\label{sec:comparison}

The constructions implement the three viability mechanisms at different
computational prices. Table~\ref{tab:constructions} records exactly where each
guarantee is proved and where an infinite-dimensional well-posedness check
remains conditional.

\begin{table}[ht]
\centering
\small
\begin{tabularx}{\textwidth}{@{}>{\raggedright\arraybackslash}X>{\raggedright\arraybackslash}X>{\raggedright\arraybackslash}Xll@{}}
\toprule
 & \textbf{No-arbitrage} & \textbf{Well-posedness} & \textbf{Factor rank} & \textbf{Interpretation} \\
\midrule
I. Price-coordinate reflected & exact, pathwise & \tier{Conditional} & full & reflection opaque \\
II. Tangent-projected KL & exact per step (grid) & QP \tier{Proved} & $q$ chosen & KL factors retained \\
III. Parameter maps (SSVI/neural) & exact on image & \tier{Proved} (fin.\ dim.) & $p$ fixed & parameters named \\
\bottomrule
\end{tabularx}
\caption{The three constructions against the trichotomy of Corollary~\ref{cor:trichotomy}. ``No-arbitrage'' describes how the static constraints are enforced; ``factor rank'' describes how the dimension of surface innovations is determined.}
\label{tab:constructions}
\end{table}

\section{Neural operators under arbitrage constraints}\label{sec:neuraloperators}

A neural operator is a parametrized map between function spaces, built from
pointwise maps and learned integral kernels. A branch--trunk representation
\begin{equation}\label{eq:branchtrunk}
G_\vartheta(v)(x)=\sum_{j=1}^{p}b_j(v)t_j(x)
\end{equation}
factors through $\R^p$ and is therefore a learned finite-rank realization
\citep{kovachki2023neural,lu2021deeponet}. Its
linear orthogonal member cannot improve on the KL reconstruction floor.

\begin{proposition}[KL optimality among linear rank-$p$ decoders;
\tier{Proved}]\label{prop:kloptimal}
Let $X$ be centered in $\Hs$ with covariance eigenvalues
$\lambda_1\ge\lambda_2\ge\cdots$. For every orthogonal projection $P_V$ onto a
$p$-dimensional subspace,
\[
\E\|X-P_VX\|_{\Hs}^2=\tr C-\tr(P_VC)\ge\sum_{j>p}\lambda_j,
\]
with equality for the leading KL eigenspace.
\end{proposition}

\begin{proof}
The identity follows from cyclicity of the trace; Ky Fan's maximum principle
gives $\tr(P_VC)\le\sum_{j\le p}\lambda_j$.
\end{proof}

Nonlinear encoders can improve on this linear floor only by using non-Gaussian
or state-dependent structure. The boundary theorem identifies exactly where
that structure is unavoidable.

\begin{corollary}[Learned Gaussian generators inherit the half-law;
\tier{Proved}]\label{cor:neuralhalflaw}
Suppose a learned one-step generator has
$w_\Delta=w_0+\Delta b_\vartheta(w_0)+\sqrt\Delta\,
\Sigma_\vartheta(w_0)Z$, $Z\sim N(0,I)$. At an active smooth constraint with
normal $\ell$, if $\Sigma_\vartheta(w_0)^*\ell\ne0$, then its violation
probability tends to $1/2$, independently of the drift and any soft penalty
used in training.
\end{corollary}

Thus exact feasibility must be architectural. In the convex price chart it is
obtained by a simplex head.

\begin{theorem}[Constrained universal approximation; \tier{Proved}]
\label{thm:simplexhead}
Let $X$ be a Banach space, $K\subset X$ closed and convex, $A$ compact, and
$F:A\to K$ continuous. For every $\varepsilon>0$ there exist anchors
$c_1,\ldots,c_J\in F(A)$ and continuous logits $\psi_j$ such that
\begin{equation}\label{eq:simplexhead}
G(v)=\sum_{j=1}^J p_j(v)c_j,\qquad
p(v)=\operatorname{softmax}(\psi_1(v),\ldots,\psi_J(v)),
\end{equation}
satisfies $G(v)\in K$ for every $v$ and
$\sup_{v\in A}\|G(v)-F(v)\|_X<\varepsilon$. If the logits are approximated by a
universal neural-functional class, the same conclusion holds up to an
arbitrarily small additional error while feasibility remains exact.
\end{theorem}

\begin{proof}
Compactness of $F(A)$ supplies a finite $\varepsilon/3$-net of feasible anchors.
A continuous partition of unity subordinate to the inverse images of the net
balls gives a feasible convex-combination approximation. Lift its weights by a
small positive constant, renormalize, and take logarithms to obtain continuous
softmax logits. Uniform approximation of those logits and uniform continuity
of softmax complete the proof.
\end{proof}

Taking $K=\Kcmbar$ yields arbitrage-free price-surface operators whose implied
variance inverse is uniformly well defined for every parameter value; taking
$K=\Kc$ remains valid for price-only outputs. The local chart gives a second
route:
\begin{equation}\label{eq:conehead}
v\longmapsto a_\vartheta(v)=\exp(G_\vartheta(v))
\longmapsto c[a_\vartheta(v)].
\end{equation}
Every output is arbitrage-free, and Proposition~\ref{prop:localinverse} makes
the construction dense and stable on continuous uniformly elliptic target
families. The price route reaches boundary anchors cheaply; the local route is
interior-only and pays for a parabolic solve, but its constraint head is global
and unconstrained in log coordinates.

\section{Arbitrage-free normalizing flow maps}\label{sec:normalizingflows}

Normalizing flows add something that a generic generator does not: an
invertible map with a tractable likelihood \citep{rezende2015flows,papamakarios2021flows}.
The word ``flow'' here refers to composition of density transforms, not to
calendar-time market dynamics. The constraint geometry dictates where the
invertible map should live. In total-variance coordinates the admissible set is
nonconvex; in local coordinates its interior is a positive cone; in price
coordinates finite convex hulls have simplex interiors. The last two admit
explicit normalizing charts.

\subsection{Local-cone flow and exact grid likelihood}\label{sec:localflow}

Fix $p$ local-variance grid nodes, a conditioning state $s$, an invertible
matrix $V\in\R^{p\times p}$, a location $\mu(s)$, and a conditional $C^1$
diffeomorphism $F_\vartheta(\cdot;s):\R^p\to\R^p$. Let
\begin{equation}\label{eq:localflowmap}
\varepsilon\sim p_0,\qquad y=F_\vartheta(\varepsilon;s),\qquad
u=\mu(s)+Vy,\qquad a=\exp(u),
\end{equation}
where the exponential is componentwise. Between grid nodes, interpolate $u$
and exponentiate, never interpolate $a$ with a sign-indefinite basis.

\begin{theorem}[Local-volatility normalizing flow; \tier{Proved}]
\label{thm:localflow}
Map \eqref{eq:localflowmap} is a diffeomorphism from $\R^p$ onto the positive
orthant. Its conditional log-density is
\begin{equation}\label{eq:localflowdensity}
\log p_A(a\mid s)=\log p_0(\varepsilon)
-\log|\det D_\varepsilon F_\vartheta(\varepsilon;s)|
-\log|\det V|-\sum_{j=1}^p\log a_j,
\end{equation}
where
$\varepsilon=F_\vartheta^{-1}(V^{-1}(\log a-\mu(s));s)$. The positive
interpolated field, composed with the inverse Dupire map of
Proposition~\ref{prop:localinverse}, produces a statically arbitrage-free price
surface for every latent draw and every parameter value.
\end{theorem}

\begin{proof}
Each component is invertible: $F_\vartheta$ by assumption, $V$ by nonsingularity,
and the exponential from $\R^p$ to $(0,\infty)^p$. The chain-rule determinant is
$|\det D F_\vartheta|\,|\det V|\prod_ja_j$; the ordinary change-of-variables
formula gives \eqref{eq:localflowdensity}. Positivity and
Proposition~\ref{prop:localinverse} give feasibility.
\end{proof}

Training in $a$-coordinates retains the exact likelihood
\eqref{eq:localflowdensity}. A likelihood assigned directly to the resulting
price grid must also include the Jacobian of the discrete parabolic solution
map; omitting it changes the statistical model.

\subsection{Invertible price-simplex flow}\label{sec:simplexflow}

The softmax head of Theorem~\ref{thm:simplexhead} is feasible but not injective:
adding a common constant to all logits changes nothing. Fix the gauge through
stick breaking. For $z\in\R^{J-1}$, let $v_j=(1+e^{-z_j})^{-1}$ and
\begin{equation}\label{eq:stickbreaking}
p_j=v_j\prod_{\ell<j}(1-v_\ell),\quad j<J,
\qquad p_J=\prod_{\ell<J}(1-v_\ell).
\end{equation}

\begin{theorem}[Price-simplex normalizing flow; \tier{Proved}]
\label{thm:simplexflow}
Fix $0<m<\Mcap<\infty$. Let $c_1,\ldots,c_J\in\Kcmbar$ be affinely independent
and let $F_\vartheta:\R^{J-1}\to\R^{J-1}$ be a $C^1$ diffeomorphism. The composition
\[
\varepsilon\mapsto z=F_\vartheta(\varepsilon)\mapsto p(z)\mapsto
c=\sum_{j=1}^Jp_jc_j
\]
is a diffeomorphism onto the relative interior of
$\operatorname{conv}\{c_1,\ldots,c_J\}\subset\Kcmbar$. Define the anchor Gram
matrix $G_{ij}=\ip{c_i-c_J}{c_j-c_J}_{\Hs}$, $i,j<J$. The density relative to
$(J-1)$-dimensional Hausdorff volume on that affine hull is
\begin{align}\label{eq:simplexflowdensity}
\log p_C(c)=\log p_0(\varepsilon)
&-\log|\det D F_\vartheta(\varepsilon)|
-\frac12\log\det G\\
&-\sum_{j=1}^{J-1}\{\log v_j+(J-j)\log(1-v_j)\}.\nonumber
\end{align}
Every draw is statically arbitrage-free, including during training, and has a
unique inverse in $\Ksurf_{m,\Mcap}$.
\end{theorem}

\begin{proof}
Stick breaking is a bijection from $\R^{J-1}$ to the open simplex; affine
independence makes barycentric coordinates unique. Its Jacobian in the first
$J-1$ coordinates is lower triangular and has log-determinant
$\sum_{j<J}[\log v_j+(J-j)\log(1-v_j)]$. The affine anchor map contributes the
constant volume factor $\sqrt{\det G}$. Multiplying these determinants with the
flow determinant proves \eqref{eq:simplexflowdensity}; convexity of $\Kcmbar$
and Proposition~\ref{prop:chartcorrespondence} prove feasibility and
invertibility back to total variance.
\end{proof}

\subsection{Conditional signature flows and the function-space boundary}
\label{sec:functionspaceflow}

The fading signature supplies the conditioning state without compromising
invertibility in the innovation:
\begin{equation}\label{eq:conditionalsigflow}
y_t=F_\vartheta\!\left(\varepsilon_t;
\mathbb S_t^{\lambda,\le M}(z)\right).
\end{equation}
Composing $y_t$ with either Theorem~\ref{thm:localflow} or
Theorem~\ref{thm:simplexflow} gives a history-dependent conditional density of
arbitrage-free surfaces. Unlike a deterministic signature readout, it separates
predictable history from irreducible innovation and can be scored by log
likelihood.

There is, however, no translation-invariant Lebesgue measure on an
infinite-dimensional Hilbert space. Consequently, the determinant formulae
above are exact for finite grids or finite-rank manifolds, not automatically
discretization-invariant function-space likelihoods. If a reference Gaussian
measure $\gamma$ on $\Hs$ is used, an infinite-dimensional transformation must
preserve quasi-invariance. For maps of the form $T=I+K$, this requires, in the
classical Ramer framework, shifts in the Cameron--Martin directions and a
Hilbert--Schmidt derivative together with integrability and invertibility
conditions; the Radon--Nikodym derivative involves a regularized Fredholm
determinant rather than $\det DT$ \citep{ramer1974nonlinear,bogachev1998gaussian}.
Generic coordinatewise neural flows need not satisfy these conditions and can
send equivalent Gaussian measures to mutually singular laws. This paper
therefore makes exact likelihood claims only in the finite-dimensional settings
of Theorems~\ref{thm:localflow}--\ref{thm:simplexflow}.

Even in finite dimensions, log scores are meaningful only relative to a common
observation coordinate and reference measure. If $x=T(y)$ is a diffeomorphism,
then
\begin{equation}\label{eq:scoringmeasure}
 \log p_X(x)=\log p_Y(T^{-1}x)+\log|\det D T^{-1}(x)|;
\end{equation}
comparing the two raw expressions without the Jacobian compares different
statistical experiments. In Section~\ref{sec:empirics}, every likelihood test
for M9 is therefore conducted on the same local-variance grid relative to
$\dd a$. Its Gaussian benchmark $G_a$ uses the identical conditioning state and
$y=\log a\sim N(m(s),\Sigma(s))$, hence
$p_A^{G_a}(a\mid s)=\phi_{m(s),\Sigma(s)}(\log a)\prod_j a_j^{-1}$ relative to
that same measure. No KL-total-variance likelihood is compared directly with
this score.

\section{Signature fields and path-dependent surface dynamics}\label{sec:signaturefields}

KL coordinates summarize the current surface; they do not summarize how it was
reached. To encode volatility hysteresis, leverage, and regime dependence, let
$z_t\in\R^d$ collect observable surface factors, spot, realized variance, and
time. For a continuous semimartingale we use \emph{Stratonovich} iterated
integrals, so the lift is geometric. A fixed lower limit would give the usual
signature but would remember the remote past without decay; a literal sliding
window would be a delay state because both integration endpoints move. We use
instead the autonomous exponentially fading signature. For a word
$I=(i_1,\ldots,i_m)$ and $\lambda\ge0$, set
\begin{equation}\label{eq:signature}
\mathbb S_t^{\lambda,I}(z)=
\int_{0<u_1<\cdots<u_m<t}
\exp\!\left[-\lambda\sum_{r=1}^m(t-u_r)\right]
\circ\dd z^{i_1}_{u_1}\cdots\circ\dd z^{i_m}_{u_m},
\qquad \mathbb S_t^{\lambda,\varnothing}=1.
\end{equation}
For $\lambda=0$ this is the fixed-inception geometric signature. For
$\lambda>0$, level-$m$ memory decays at rate $m\lambda$. At the
full-signature level, time augmentation removes the tree-like ambiguity
relevant to forecasting; a fixed level-$M$ truncation remains a finite
approximation and is not itself injective on arbitrary paths.

\begin{definition}[Signature field]\label{def:signaturefield}
A level-$M$ fading-signature field is a surface-valued map
\begin{equation}\label{eq:sigfield}
w_t(x)=\Phi_M\!\left(x,\mathbb S^{\lambda,\le M}_{t}(z)\right),
\qquad x=(k,\tau)\in\Dom,
\end{equation}
where $\Phi_M$ is continuous (linear, kernel, or neural) in the truncated
signature coordinates. It is arbitrage-free by construction when
$\Phi_M(\Dom,\cdot)$ takes values in $\Ksurf$.
\end{definition}

\begin{proposition}[Finite-dimensional controlled realization;
\tier{Proved}]\label{prop:sigrealization}
For fixed $M$ and $d$, the truncated signature has
$N_M=\sum_{m=0}^M d^m$ tensor coordinates and satisfies the triangular
autonomous Stratonovich system
\begin{equation}\label{eq:sigcontrolled}
\dd\mathbb S^{\lambda,i_1\cdots i_m}_t
=-m\lambda\mathbb S^{\lambda,i_1\cdots i_m}_t\dd t
+\mathbb S^{\lambda,i_1\cdots i_{m-1}}_t\circ\dd z^{i_m}_t,
\qquad \mathbb S^{\lambda,\varnothing}_t=1.
\end{equation}
Consequently, if $z$ is a semimartingale (respectively a geometric rough path),
every smooth signature field is a finite-dimensional semimartingale
(respectively controlled rough-path) realization in $\Hs$.
\end{proposition}

\begin{proof}
The coordinate count is the dimension of the truncated tensor algebra.
Differentiating \eqref{eq:signature} moves the upper endpoint, producing the
prefix term, and differentiates $m$ exponential factors, producing
$-m\lambda\mathbb S^{\lambda,I}\dd t$. This proves
\eqref{eq:sigcontrolled}. Composition with smooth $\Phi_M$ gives the stated
realization by the Stratonovich or rough-chain rule. In It\^o coordinates the
usual quadratic-covariation correction must be added.
\end{proof}

The key benefit is approximation of path functionals, not a repeal of
constraints. Universality of full signatures on compact sets of time-augmented
geometric paths implies that continuous history-dependent surface rules can be
approximated by increasing $M$ and a sufficiently expressive readout; on a
finite horizon the deterministic exponential reweighting in
\eqref{eq:signature} is invertible before truncation and therefore preserves
full-signature point separation
\citep{hambly2010uniqueness,chevyrev2022signature}. Exact
arbitrage is then imposed through either convex chart; for the price head fix
$J$ feasible anchors $c_1,\ldots,c_J$:
\begin{align}
\text{price signature head:}\quad &
c_t=\sum_{j=1}^J\operatorname{softmax}_j
\bigl(L\mathbb S^{\lambda,\le M}_{t}\bigr)c_j,\quad c_j\in\Kc,
\label{eq:pricesighead}\\
\text{local signature head:}\quad &
a_t=\exp\bigl(\Phi_M(\mathbb S^{\lambda,\le M}_{t})\bigr),\qquad
c_t=c[a_t].\label{eq:localsighead}
\end{align}
Both are feasible for every history and every fitted coefficient. The first is
a path-dependent simplex realization; the second is a path-dependent local
volatility codebook. A direct total-variance readout remains subject to the
same tangent-normal conditions as any other generator.

Signature truncation and KL truncation address orthogonal dimensions of model
risk: KL truncates the \emph{surface} direction, while signatures truncate the
\emph{history} direction. A joint model
\begin{equation}\label{eq:jointtruncation}
w_t=\bar w+\sum_{i=1}^{q}e_i\,
f_i\!\left(\mathbb S^{\lambda,\le M}_{t}(z)\right)
\end{equation}
therefore has two auditable approximation axes $(q,M)$. Feasibility should be
imposed after the joint readout, not inferred from either truncation.

\begin{theorem}[Surface--memory error decomposition; \tier{Proved}]
\label{thm:twowayerror}
Let $c$ be a square-integrable $\Kc$-valued random price surface, let
$\bar c=\E c$, and let $P_q$ project onto the first $q$ eigenfunctions of the
covariance of $c-\bar c$. Suppose a level-$M$ signature readout
$\widehat y_{q,M}$ takes values in $\operatorname{ran}P_q$ and satisfies
\[
\E\|P_q(c-\bar c)-\widehat y_{q,M}\|_{\Hs}^2\le\varepsilon_{q,M}^2.
\]
Then the raw joint approximation $\widetilde c_{q,M}=\bar c+\widehat y_{q,M}$
obeys the exact orthogonal decomposition
\begin{equation}\label{eq:twowayerror}
\E\|c-\widetilde c_{q,M}\|_{\Hs}^2
=\sum_{i>q}\lambda_i
+\E\|P_q(c-\bar c)-\widehat y_{q,M}\|_{\Hs}^2
\le\sum_{i>q}\lambda_i+\varepsilon_{q,M}^2.
\end{equation}
Moreover the convex repair $c^\sharp_{q,M}=\Pi_{\Kc}\widetilde c_{q,M}$ is
arbitrage-free and has no larger error:
\begin{equation}\label{eq:projectionimproves}
\E\|c-c^\sharp_{q,M}\|_{\Hs}^2
\le\E\|c-\widetilde c_{q,M}\|_{\Hs}^2.
\end{equation}
\end{theorem}

\begin{proof}
The omitted component $(I-P_q)(c-\bar c)$ is orthogonal to both terms in
$\operatorname{ran}P_q$, giving \eqref{eq:twowayerror}; the first expectation
is the KL tail. Since $\Kc$ is closed and convex, its metric projection exists,
is unique, and satisfies $\|\Pi_{\Kc}x-y\|\le\|x-y\|$ for every $y\in\Kc$.
Take $y=c$ pointwise and integrate to obtain \eqref{eq:projectionimproves}.
\end{proof}

\begin{figure}[htbp]
\centering
\includegraphics[width=0.98\textwidth]{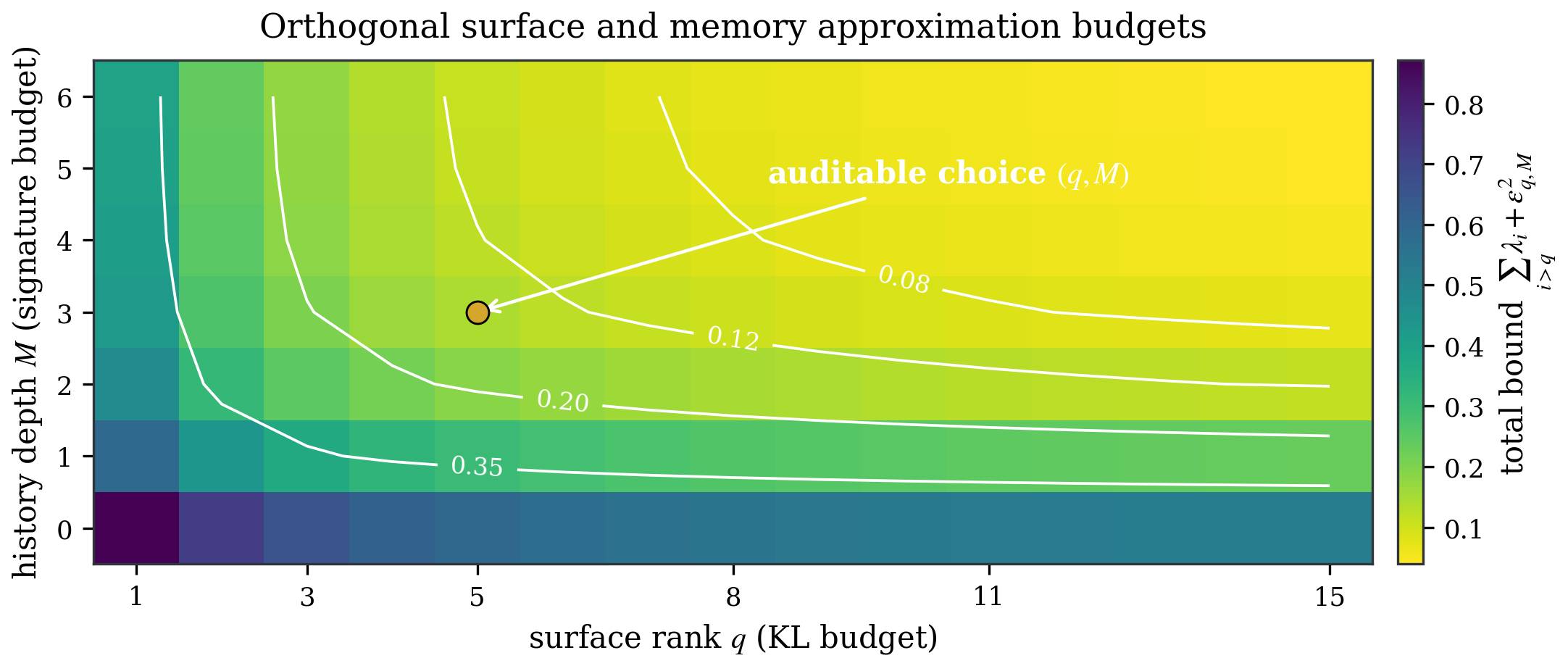}
\caption{The two approximation axes of Theorem~\ref{thm:twowayerror}. The
illustration uses normalized $\lambda_i\propto i^{-2}$ and
$\varepsilon_{q,M}^2=0.48e^{-0.9M}$ to display the bound
$\sum_{i>q}\lambda_i+\varepsilon_{q,M}^2$. Rank $q$ controls omitted surface
directions; depth $M$ controls path-memory approximation. The values are
illustrative and are not model estimates.}
\label{fig:memorybudget}
\end{figure}

Theorem~\ref{thm:twowayerror} turns model selection into a two-dimensional
budget: increase $q$ until the spectral tail is acceptable, then increase $M$
and readout capacity until the history error is acceptable. Constraint repair
cannot worsen the Hilbert error in price coordinates; a simplex signature head
achieves feasibility without the repair step.

\section{A falsifiable empirical program}\label{sec:empirics}

This section is a pre-registered design. It fixes data, models, metrics, and hypotheses before any estimation; thresholds named below are design constants, not results. No numbers in this section are estimates.

\subsection{Data and surface construction}\label{sec:data}

End-of-day SPX option quotes, January 2005 -- December 2025, covering at least two stress episodes (2008--09, 2020) and extended calm regimes. Filters: positive bids, midquotes from bid--ask, maturities $\tau\in[7/365,\,2]$ years, forward log-moneyness $k\in[-1.5,\,0.5]$, standard OptionMetrics-style screens for stale and crossed quotes. Forwards and discounting from put--call parity per maturity. Two construction pipelines, to control for construction artifacts: (a) arbitrage-free smoothing in the sense of \citet{fengler2009smoothing}; (b) arbitrage repair in price space in the sense of \citet{cohen2020repair}, followed by interpolation. Common evaluation grid: $41\times21$ points on $\Dom$; Sobolev weight $\rho\equiv1$ for headline results, vega-weighting as robustness. Increments computed at fixed $(k,\tau)$ on the grid; fixed-strike variants as robustness (the calendar constraint distinguishes them, Section~\ref{sec:families}).

\subsection{Models M1--M9}\label{sec:models}

\begin{enumerate}[label=\textbf{M\arabic*.},leftmargin=2.6em]
\item Unconstrained Gaussian KL field: Section~\ref{sec:spectral} factors, Gaussian innovations, no projection. (The straw man that Theorem~\ref{thm:boundary} says must fail at the boundary --- included to measure \emph{how}.)
\item Tangent-projected KL field: Construction II with the same factors, slack schedule $\epsilon_j$ per Section~\ref{sec:constructionII}.
\item Price-coordinate reflected field: Construction I discretized by projected Euler in $c$-coordinates.
\item SSVI diffusion: Construction III with \citet{gatheral2014arbitrage} sufficient conditions.
\item eSSVI diffusion: Construction III per \citet{hendriks2019essvi}.
\item Neural-SDE market model in the sense of \citet{cohen2023neural}, authors' recommended configuration.
\item Generative surface simulator in the sense of \citet{cont2023simulation}, authors' recommended configuration.
\item Fading-signature-field generator: the same observable state as M2, augmented by spot and realized variance history, with pre-registered $(q,M,\lambda)$ and the price-simplex head \eqref{eq:pricesighead}.
\item Conditional normalizing-flow field: the M8 fading-signature state conditions the local-cone flow \eqref{eq:conditionalsigflow}, trained by the exact grid likelihood \eqref{eq:localflowdensity} and mapped to prices through Dupire.
\end{enumerate}
All models are fit on a common training window and evaluated out of sample on rolling windows; splits, seeds, and configurations are fixed in the registered protocol (Appendix~\ref{app:checklist}).

\subsection{Metric families F1--F6}\label{sec:metrics}

\begin{enumerate}[label=\textbf{F\arabic*.},leftmargin=2.6em]
\item \emph{Arbitrage integrity.} (a) violation frequency: share of simulated steps with $\operatorname{ess\,inf} g<0$, calendar violation, or floor breach on the grid; (b) violation severity: distribution of $\operatorname{ess\,inf} g$ conditional on violation; (c) independent rerun of the interval certificate in Proposition~\ref{prop:convexity}(ii), plus a wider search for lower-complexity rational certificates and a separate search requiring the certified violating point to lie inside the empirical window $k\in[-1.5,0.5]$.
\item \emph{Increment distribution.} Per-mode innovation distributions against Gaussian benchmarks (tails, skewness); spectrum $\{\lambda_i\}$ shape and stability across calm/stress subsamples.
\item \emph{Dynamic consistency.} Autocorrelation of level/skew/curvature mode scores; skew-stickiness behavior; regime transfer (fit calm, evaluate stress, and conversely), with rough/path-dependent benchmarks \citep{gatheral2018rough,guyon2023path} as descriptive references.
\item \emph{Hedging.} Realized out-of-sample vega P\&L variance of test books hedged by (a) bucket regression and (b) the field hedge \eqref{eq:alphastar} with $\widehat C$ estimated on the training window; the audit statistic is the realized-to-lower-bound ratio against \eqref{eq:residual}.
\item \emph{Boundary behavior.} Conditional on small pre-step margin $\delta$ (empirical proxy: grid $\operatorname{ess\,inf} g$ in the lowest decile), one-step violation frequency by model, against the standardized Gaussian boundary-layer curve \eqref{eq:uniformboundarylayer}, the $1/2$ limit of Theorem~\ref{thm:boundary}, and the exponential interior bound \eqref{eq:borell}.
\item \emph{Distributional calibration.} Out-of-sample conditional log score for models with tractable densities; Rosenblatt/PIT diagnostics in each model's declared observation coordinate; energy score for the full surface; and left-tail coverage of vega-weighted surface P\&L. The M9 likelihood and its Gaussian benchmark $G_a$ are both densities of the same local-variance vector relative to $\dd a$, with the same conditioning state. Likelihoods in KL total-variance coordinates are reported separately and are never ranked against local-coordinate scores, as required by \eqref{eq:scoringmeasure}.
\end{enumerate}

\subsection{Hypotheses}\label{sec:hypotheses}

\begin{enumerate}[label=\textbf{H\arabic*.},leftmargin=2.6em]
\item \textbf{Boundary law.} M1's one-step violation frequency, conditional on active-margin states (F5), exceeds $0.25$ and follows the standardized-margin curve $\overline N(\Phi/(\sqrt\Delta\sigma_\Phi))$ of Corollary~\ref{cor:boundarylayer}, increasing toward $0.5$ as the conditioning margin and simulation step shrink; M2--M9 remain below $0.01$ on the same conditioning set. Failure of the M1 curve or boundary limit falsifies the empirical relevance of Theorem~\ref{thm:boundary} at daily steps.
\item \textbf{Interior equivalence.} On interior states (top half of the margin distribution), M1 and M2 increment laws are statistically indistinguishable at the pre-registered test level across F2 statistics; i.e., projection matters only near the boundary, as Proposition~\ref{prop:survival} predicts.
\item \textbf{Spectral parsimony.} $q^\ast=5$ KL modes capture at least $90\%$ of increment variance ($\sum_{i\le5}\lambda_i\ge0.9\tr C$) in both calm and stress subsamples, with mode \emph{shapes} stable across regimes (subspace angle below the registered threshold) while mode \emph{scales} are not.
\item \textbf{Hedging value.} The field hedge \eqref{eq:alphastar} reduces out-of-sample vega P\&L variance relative to bucket regression by at least $10\%$ on the registered test books; the improvement concentrates in stress windows.
\item \textbf{Manifold rigidity.} For M4--M5, the residual field $w_t-\varphi(\widehat\theta_t)$ exhibits significant autocorrelation in at least one mode score at the registered level, i.e., low-dimensional parameter maps leave dynamically structured curvature on the table.
\item \textbf{Value of path memory.} At matched surface rank $q$, the primary attribution test compares M8 with $M\ge1$ against an M8$(0)$ ablation in which the positive-level signature block is replaced by a zero vector of the same dimension while the price-simplex head, downstream architecture, optimizer, training schedule, and trainable-parameter budget are held fixed. M8 must reduce one-step out-of-sample Sobolev RMSE by at least $5\%$ and improve calm-to-stress transfer at the registered level. M2 remains a separate memoryless benchmark, not an identity claim. A separate $M=1$ versus $M>1$ ablation tests whether iterated interactions, rather than exponentially weighted first-level history alone, carry the gain.
\item \textbf{Value of conditional density.} On the common local-variance grid and conditioning state, M9 improves the paired out-of-sample conditional log score over $G_a$, the Gaussian model for $\log a$ defined after \eqref{eq:scoringmeasure}, at the registered level and passes the joint Rosenblatt/PIT uniformity test after multiplicity correction. Failure of either condition rejects the claim that invertible flow flexibility is calibrated rather than merely expressive.
\end{enumerate}
Each hypothesis names its own falsification: the design is symmetric in what it can lose.

\section{Scope and theorem frontier}\label{sec:scope}

The proved statements have six deliberate boundaries. First, $\tau_{\min}>0$
and finite $k$ exclude the expiry singularity and global wings; no claim here
replaces small-time asymptotics or the moment constraints of
\citet{lee2004moment}. Second, Proposition~\ref{prop:localinverse} is uniformly
elliptic and therefore does not cover the zero--pole strata of
Corollary~\ref{cor:dupiredivisor}. Third, the ordinary normalizing-flow
determinants are finite-grid or finite-rank results; genuine function-space
likelihoods require the quasi-invariance conditions stated in
Section~\ref{sec:functionspaceflow}. Fourth, the surface constructions enforce
static feasibility; a pricing model must additionally satisfy
Proposition~\ref{prop:musiela}, including stochastic-forward and discount-factor
terms. Fifth, the hedge is a frictionless, first-order, one-period covariance
hedge, and Section~\ref{sec:empirics} is a frozen protocol rather than an
empirical result. Sixth, Theorem~\ref{thm:gridreflection} proves reflected
well-posedness and numerical convergence on each fixed quote grid, while the
specific infinite-dimensional reflected equation remains conditional in
Proposition~\ref{prop:reflectedstatus} and the native total-variance limit in
Proposition~\ref{prop:qpstatus}(iii) remains conjectural. Expiry asymptotics, global-wing geometry, degenerate local
ellipticity, dynamically risk-neutral learned generators, transaction costs,
and execution-aware hedging are separate extensions rather than hidden
assumptions.

\section{Conclusion}\label{sec:conclusion}

The volatility surface is a stochastic field that lives against a wall. The wall is visible in three coordinate systems: a nonconvex total-variance set, a convex price set, and a positive local-variance cone. Dupire's ratio shows that the calendar and butterfly boundaries become the zeros and poles of the local field. The boundary half-law proves that unconstrained Gaussian dynamics cannot be exact there; the viable exits are tangency, reflection, or a constrained realization. Musiela transport then marks the second layer of consistency: staying inside the wall prevents static arbitrage, while fixed-contract martingales impose the separate dynamic drift restriction. Inside the set, operator covariance supplies spectral factors, truncation error, and the metric for the optimal vega-field hedge. Neural operators do not escape this geometry: soft penalties inherit the half-law, while simplex and cone heads achieve exact feasibility. Normalizing flows add likelihood without sacrificing feasibility when invertibility is built inside those charts; their ordinary determinants remain finite-dimensional and coordinate-measure specific, while genuine function-space likelihoods require Gaussian quasi-invariance. Fading-signature fields add memory and condition the flow without changing the constraint rule. The resulting theory separates two approximation budgets---surface rank and history depth---and three statistical tasks---reconstruction, conditional prediction, and calibrated density estimation. Reflected price dynamics are fully well posed and numerically convergent on every fixed quote grid; the exact function-space verification and the native nonconvex tangent-projection limit remain open. Constructing flexible learned generators satisfying the full fixed-contract drift restriction is the other central frontier, and the empirical program is designed to determine which refinements matter at market time scales.

\appendix

\section{Directional derivative and remainder for the butterfly functional}\label{app:Dg}

Fix $w\in\Hs$ with $m\le w\le M_0$ on $\overline{\Dom}$ and $h\in\Hs$. Write $A=1-\frac{k w_k}{2w}$ and $g[w]=A^2-\frac{w_k^2}{4}\bigl(\frac1w+\frac14\bigr)+\frac{w_{kk}}{2}$ as in \eqref{eq:gfun}.

\paragraph{Derivative.} Differentiating each term of $g[w+\varepsilon h]$ at $\varepsilon=0$:
\[
\frac{\dd}{\dd\varepsilon}\Big|_0 A_\varepsilon
= -\frac{k}{2}\cdot\frac{h_k w - w_k h}{w^2},
\qquad
\frac{\dd}{\dd\varepsilon}\Big|_0\Bigl[-\frac{(w_k+\varepsilon h_k)^2}{4}\Bigl(\frac{1}{w+\varepsilon h}+\frac14\Bigr)\Bigr]
= -\frac{w_k h_k}{2}\Bigl(\frac1w+\frac14\Bigr)+\frac{w_k^2}{4w^2}h,
\]
and the linear term contributes $h_{kk}/2$. Collecting, and using $2A\cdot\bigl(-\frac{k}{2w^2}\bigr)(h_kw-w_kh) = -\frac{kA}{w^2}(wh_k-w_kh)$, gives \eqref{eq:Dg}.

\paragraph{Remainder.} Let $u=w+\theta h$, $\theta\in[0,1]$, and suppose $\norm{h}$ small enough that $u\ge m/2$. Each term of $g$ is a polynomial in $(k,\,u^{-1},\,u_k,\,u_{kk})$ of joint degree at most four with $u_{kk}$ appearing linearly. Taylor with integral remainder in $\varepsilon$ gives
\begin{equation}\label{eq:remainder}
g[w+h] - g[w] - Dg[w]h \;=\; \mathcal{Q}[w;h],
\qquad
\norm{\mathcal{Q}[w;h]}_{L^1(\Dom)} \;\le\; \kappa\bigl(m,M_0,\norm{w}_{\Hs}\bigr)\,\norm{h}_{\Hs}^2,
\end{equation}
because every second-order-in-$h$ term is a product of at most two factors from $\{h,h_k\}$ and bounded coefficients, controlled by $\norm{h}_{C^0}\norm{h_k}_{L^2}$, $\norm{h_k}_{L^4}^2$, or $\norm{h}_{C^0}\norm{h_{kk}}_{L^2}$ --- each $\lesssim\norm{h}^2_{\Hs}$ by Proposition~\ref{prop:embedding} --- while no term is quadratic in $h_{kk}$ (linearity of $g$ in the top order). Polarizing the second variation and applying the same estimates to coefficient differences gives the Hessian bound and continuity stated explicitly in Lemma~\ref{lem:PhiC2}. If moreover $w,h\in C^2(\overline{\Dom})$ then term-by-term the same bound holds in $L^\infty$ with $\norm{h}^2_{C^2}$ on the right, which is the form used in Proposition~\ref{prop:feasible}.

\section{The joint spot--field hedge}\label{app:joint}

Extend the state to $(s,w)\in\R\oplus\Hs$ with $s$ the log-forward, joint increment covariance
\[
\mathcal{C} \;=\;
\begin{pmatrix}
c_{ss} & c_{sw}^\ast \\[2pt]
c_{sw} & C
\end{pmatrix},
\qquad c_{ss}\ge0,\quad c_{sw}\in\Hs,
\]
book sensitivity $(\partial_s\Pi,\ \nu)$, and instruments with sensitivities $(d_i,\ \eta_i)$, $i=1,\dots,n$, stacked into $\mathsf{H}:\R^n\to\R\oplus\Hs$. The one-period variance of the hedged book is the same quadratic \eqref{eq:variance} with $(\nu,H,C)$ replaced by $((\partial_s\Pi,\nu),\mathsf{H},\mathcal{C})$, so Theorem~\ref{thm:hedge} applies verbatim:
\[
\alpha^\ast_{\mathrm{joint}} \;=\; \bigl(\mathsf{H}^\ast\mathcal{C}\mathsf{H}\bigr)^{-1}\mathsf{H}^\ast\mathcal{C}
\begin{pmatrix}\partial_s\Pi\\ \nu\end{pmatrix}.
\]
Writing the blocks out, the cross-covariance $c_{sw}$ --- the spot--surface leverage field, empirically dominated by the short-dated skew reaction --- is what shifts option hedge weights away from the pure-field solution \eqref{eq:alphastar} and reallocates part of the vega hedge into the delta line. Setting $c_{sw}=0$ decouples the blocks and recovers \eqref{eq:alphastar} together with the classical minimum-variance delta $\alpha_{\mathrm{spot}}=\partial_s\Pi/1$ against a unit-delta instrument. The formula also exhibits the familiar practitioner fact that ``sticky'' conventions are statements about $c_{sw}$, not about $\nu$.

\section{Reproducibility checklist}\label{app:checklist}

To be published with the empirical implementation, before estimation:
\begin{enumerate}[label=(\arabic*),leftmargin=2.2em]
\item Data vendor, extraction date, and the exact filter cascade with per-filter attrition counts.
\item Both surface-construction pipelines (smoothing and repair) with library versions and settings; the evaluation grid; the weight $\rho$.
\item Model configurations M1--M9 and $G_a$: factor count $q$, signature depth $M$, decay $\lambda$, normalizing-flow architecture and base law, observation coordinate and reference measure, slack schedule $\epsilon_j$, parameter families, network architectures and seeds, training windows. Record the exact M8$(0)$ masking rule and matched parameter budget.
\item Rolling-window scheme: training/evaluation spans, step, and the calm/stress subsample definitions by date.
\item Test levels and all thresholds of H1--H7 as registered in Section~\ref{sec:hypotheses}, frozen before data contact.
\item Code and environment: repository hash, container image, hardware; one-command reproduction of every table.
\item The certificate search protocol for F1(c), including the SVI parameter box searched and separate declaration rules for a global-window certificate and a violation certified inside the empirical window $k\in[-1.5,0.5]$.
\end{enumerate}

\bibliographystyle{plainnat}
\bibliography{references}

\end{document}